\documentclass[11pt]{article}

\usepackage[letterpaper,margin=1.00in]{geometry}
\usepackage{amsmath, amssymb, amsthm, thmtools, amsfonts, dsfont}
\usepackage{bbm}

\usepackage{ifthen}

\usepackage{tikz}
\usetikzlibrary{positioning,decorations.pathreplacing}

\usepackage{cite}
\usepackage{appendix}
\usepackage{graphicx}
\usepackage{color}
\usepackage{algorithm}
\usepackage[noend]{algpseudocode}
\usepackage{epstopdf}
\usepackage[textsize=tiny]{todonotes}

\usepackage{framed}
\usepackage[framemethod=tikz]{mdframed}
\usepackage[bottom]{footmisc}
\usepackage[shortlabels]{enumitem}
\setitemize{noitemsep,topsep=3pt,parsep=3pt,partopsep=3pt}
\usepackage[font=small]{caption}
\usepackage{xspace}

\usepackage{mathtools}

\newtheorem{theorem}{Theorem}[section]
\newtheorem{lemma}[theorem]{Lemma}
\newtheorem{meta-theorem}[theorem]{Meta-Theorem}

\newtheorem{remark}[theorem]{Remark}
\newtheorem{corollary}[theorem]{Corollary}

\newtheorem{definition}[theorem]{Definition}

\newcommand{\FullOrShort}{full}

\ifthenelse{\equal{\FullOrShort}{full}}{
    
  \newcommand{\fullOnly}[1]{#1}
  \newcommand{\shortOnly}[1]{}
  }{
    \newcommand{\fullOnly}[1]{}
    \newcommand{\shortOnly}[1]{#1}
  }

\definecolor{darkgreen}{rgb}{0,0.5,0}
\usepackage{hyperref}
\hypersetup{
    unicode=false,          % non-Latin characters in Acrobat’s bookmarks
    colorlinks=true,        % false: boxed links; true: colored links
    linkcolor=red,          % color of internal links (change box color with linkbordercolor)
    citecolor=darkgreen,        % color of links to bibliography
    filecolor=magenta,      % color of file links
    urlcolor=cyan           % color of external links
}
\usepackage[capitalize, nameinlink]{cleveref}
\crefname{theorem}{Theorem}{Theorems}
\Crefname{lemma}{Lemma}{Lemmas}
\Crefname{observation}{Observation}{Observations}
\Crefname{equation}{}{}

\algnewcommand\algorithmicswitch{\textbf{switch}}
\algnewcommand\algorithmiccase{\textbf{case}}

\algdef{SE}[SWITCH]{Switch}{EndSwitch}[1]{\algorithmicswitch\ #1\ \algorithmicdo}{\algorithmicend\ \algorithmicswitch}%
\algdef{SE}[CASE]{Case}{EndCase}[1]{\algorithmiccase\ #1}{\algorithmicend\ \algorithmiccase}%
\algtext*{EndSwitch}%
\algtext*{EndCase}%
\usepackage{thmtools} 
\usepackage{thm-restate}

\usepackage[textsize=tiny]{todonotes}

\usepackage[capitalize, nameinlink]{cleveref}
\crefname{theorem}{Theorem}{Theorems}
\Crefname{lemma}{Lemma}{Lemmas}
\Crefname{observation}{Observation}{Observations}
\Crefname{equation}{}{}

\newcommand{\eps}{\varepsilon}

\newcommand{\card}[1]{\left| #1 \right|}

\newcommand{\eqdef}{{\stackrel{\rm def}{=}}}
\newcommand{\poly}{{\rm poly}}

\newcommand{\floor}[1]{\left\lfloor #1 \right\rfloor}
\newcommand{\ceil}[1]{\left\lceil #1 \right\rceil}

\newcommand{\dist}{\mathop{\mbox{\rm dist}}\nolimits}

\providecommand{\E}{{\rm \mathbb{E}}}
\providecommand{\Var}{{\rm \mathbb{V}ar}}
\providecommand{\Cov}{{\rm Cov}}

\newcommand*{\congested}{\textsf{CONGESTED-CLIQUE}}
\newcommand*{\mpc}{\textsf{MPC}}

\newcommand*{\local}{\textsf{LOCAL}}
\newcommand*{\congest}{\textsf{CONGEST}}

\newcommand*{\mis}{\textsf{MIS}}
\newcommand*{\mm}{\textsf{MM}}

\providecommand{\matching}{\mathcal{M}}
\providecommand{\indepset}{\mathcal{I}}
\providecommand{\machines}{\mathbf{M}}
\providecommand{\memory}{\mathbf{S}}

\providecommand{\spaceconst}{\alpha}
\providecommand{\localspace}{O(n^\alpha)}
\providecommand{\arb}{\lambda}
\providecommand{\hpart}{\mathcal{H}}
\providecommand{\last}{L}

\newcommand{\goodset}[1]{T(#1)}
\providecommand{\degreehigh}{\beta}

\newcommand{\Shell}{\mathcal{S}}
\newcommand{\Shellgood}{\mathcal{S}_{\textsf{Good}}}
\newcommand{\highdeg}{\sqrt{\Delta_i}}
\newcommand{\betadeg}{\Delta^{0.03}}
\newcommand{\pvalue}{\beta^{3}}

\renewcommand{\paragraph}[1]{\vspace{0.15cm}\noindent {\bf #1}:}

\usepackage{setspace}

\usepackage{mathtools}
\DeclarePairedDelimiter{\abs}{\lvert}{\rvert}

\makeatletter
\let\oldabs\abs
\def\abs{\@ifstar{\oldabs}{\oldabs*}}
\makeatother

{

    \newcommand{\IncludePictures}[1]{}
   
  }

\begin{document}

\date{}

\title{\huge Deterministic Massively Parallel Symmetry Breaking\\ for Sparse Graphs}

\author{
    Manuela Fischer \\
    \small{ETH Zurich}\\
    \small{\texttt{manuela.fischer@inf.ethz.ch}}
    \and
    Jeff Giliberti \\
    \small{ETH Zurich}\\
    \small{\texttt{jeff.giliberti@inf.ethz.ch}}
    \and
    Christoph Grunau\thanks{\scriptsize{Supported by the European Research Council (ERC) under the European Unions Horizon 2020 research and innovation programme (grant agreement No.~853109).}} \\
    \small{ETH Zurich}\\
    \small{\texttt{cgrunau@inf.ethz.ch}}
}

\date{}

\maketitle

\setcounter{page}{0}
\thispagestyle{empty}

\begin{abstract}
We consider the problem of designing deterministic graph algorithms for the model of Massively Parallel Computation (\mpc) that improve with the sparsity of the input graph, as measured by the standard notion of arboricity. For the problems of maximal independent set (\mis), maximal matching (\mm), and vertex coloring, we improve the state of the art as follows. Let $\arb$ denote the arboricity of the $n$-node input graph with maximum degree $\Delta$.

\paragraph{\mis\ and \mm} We develop a low-space \mpc\ algorithm that \emph{deterministically} reduces the maximum degree to $\poly(\arb)$ in $O(\log \log n)$ rounds, improving and simplifying the randomized $O(\log \log n)$-round $\poly(\max(\arb, \log n))$-degree reduction of Ghaffari, Grunau, Jin~[DISC'20].
Our approach when combined with the state-of-the-art $O(\log \Delta + \log \log n)$-round algorithm by Czumaj, Davies, Parter~[SPAA'20, TALG'21] leads to an improved deterministic round complexity of $O(\log \arb + \log \log n)$. 

The above \mis\ and \mm\ algorithm however works in the setting where the global memory allowed, i.e., the number of machines times the local memory per machine, is superlinear in the input size. 
We extend them to obtain the first low-space \mis\ and \mm\ algorithms that work with linear global memory. Specifically, we show that both problems can be solved in deterministic time $O(\log \arb \cdot \log \log_{\arb} n)$, and even in $O(\log \log n)$ time for graphs with arboricity at most $\log^{O(1)} \log n$. In this setting, only a $O(\log^2 \log n)$-running time bound for trees was known due to Latypov and Uitto~[ArXiv'21].

\paragraph{Vertex Coloring} We present a $O(1)$-round deterministic algorithm for the problem of $O(\arb)$-coloring in the linear-memory regime of \mpc, with relaxed global memory of $n \cdot \poly(\arb)$.
This matches the round complexity of the state-of-the-art randomized algorithm by Ghaffari and Sayyadi~[ICALP'19] and significantly improves upon the deterministic $O(\arb^{\eps})$-round algorithm by Barenboim and Khazanov~[CSR'18].
Our algorithm solves the problem after just \textit{one} single graph partitioning step, in contrast to the involved local coloring simulations of the above state-of-the-art algorithms. 

Using $O(n + m)$ global memory, we derive a $O(\log \arb)$-round algorithm by combining the constant-round $(\Delta+1)$-list-coloring algorithm by Czumaj, Davies, Parter~[PODC'20, SIAM J. Comput.'21] with that of Barenboim and Khazanov.
\end{abstract}

\newpage
\section{Introduction}
One of the most fundamental tasks in distributed networks are symmetry breaking problems such as electing a set of leaders or resolving conflicts for the allocation of shared resources. In recent years, the ever-increasing amount of data rendered traditional graph algorithms inefficient, or even inapplicable. 
Inspired by popular large-scale computing frameworks such as MapReduce~\cite{MapReduce} and Spark~\cite{Spark}, the Massively Parallel Computation (\mpc) model~\cite{FMS+10, KSV10} has emerged as a standard theoretical abstraction for the algorithmic study of large-scale graph processing.

In the sublinear and linear regime of \mpc, there has recently been a plethora of work on fundamental combinatorial and optimization problems on graphs (see, e.g.,~\cite{BKM20, MPC-MIS-log2logn, 8948671, coyconnectivity2021, MPC-MIS-det-general, CDP20ccb, GGKMR18, doi:10.1137/1.9781611975482.99, KPP20}). As many real-world graphs are believed to be sparse, it would be desirable to develop algorithms that improve with the sparsity of graphs~\cite{CLM+18}. An established measure of sparsity is the \textit{arboricity} $\arb$ of a graph: It captures a \textit{global} notion of sparsity, as it does not impose any \textit{local} constraints such as bounded maximum degree. 
Over the last decade, a series of \textit{arboricity-dependent} algorithms~\cite{BE10, BE11, BEPS16, BK18, MPC-MIS-log3logn, MPC-MIS-loglogn, GL17, GS19, MPC-MIS-det-log2logn} has shown that for sparse graphs there exist far more efficient algorithms than for general graphs. 

In this paper, we focus on the \textit{deterministic} complexity of maximal independent set (\mis), maximal matching (\mm), and vertex coloring, parameterized by the graph arboricity, although without imposing any upper bound on the value of $\arb$. 

% MPC details
\paragraph{The \mpc\ Model} In the \mpc\ model, the distributed network consists of $\machines$ machines with memory $\memory$ each. 
The input is distributed across the machines and the computation proceeds in synchronous rounds. In each round, each machine performs an arbitrary \textit{local computation} and then communicates up to $\memory$ data with the other machines. All messages sent and received by each machine in each round have to fit into the machine’s local memory. 
The main complexity measure of an algorithm is its \textit{round complexity}, that is, the number of rounds needed by the algorithm to solve the problem. Secondary complexity measures are the \textit{global memory} usage---i.e., the number of machines times the local memory per machine---as well as the \textit{total computation} required by the algorithm, i.e., the (asymptotic) sum of all local computation performed by each machine.

% MPC regimes
For the \mis\ and \mm\ problems, we focus on the design of fully scalable graph algorithms in the \textit{low-memory} \mpc\ model (a.k.a. low-space \mpc), where each machine has strongly sublinear memory. More precisely, an input graph  $G = (V,E)$ with $n$ vertices and $m$ edges is distributed arbitrarily across machines with local memory $\memory = O(n^\spaceconst)$ each, for  constant $0 < \spaceconst < 1$, so that the global space is $\memory_{Global} = \Omega(n + m)$. 

For vertex coloring, we consider the \textit{linear-memory} \mpc\ model, where machines have local memory $\memory = O(n)$. This model is closely linked to the \congested\ model~\cite{Congested}, where the communication graph is complete but the  message size is restricted to $O(\log n)$ bits. By a result of Lenzen~\cite{L13}, linear \mpc\ results can often be extended to \congested~\cite{BDH18}.

\paragraph{Maximal Independent Set and Maximal Matching}
These two problems have been studied extensively since the very beginning of the area \cite{ABI86, BEPS16, Gha16, II86, L85} and are defined as follows. A \textit{maximal independent set} is a set $\indepset  \subseteq V$ such that no two vertices in $\indepset$ are adjacent and every vertex $u \in V \setminus \indepset$ has a neighbor in $\indepset$. A \textit{maximal matching} is a set $\matching  \subseteq E$  such that no two edges in $\matching$ are adjacent and every edge $e \in V \setminus \matching$ has a neighbor in $\matching$.

\textbf{Randomized \mis\ and \mm} One of the recent highlights in low-memory \mpc\ has been a sequence of fast algorithms for \textit{bounded-arboricity graphs}. The authors in~\cite{MPC-MIS-log3logn} were the first to break the linear memory barrier devising an $O(\log^3 \log n)$-time algorithm for trees. This result was improved by~\cite{MPC-MIS-log2logn} who gave a general $O(\log^2 \log n)$-time algorithm that reduces the maximum degree to $\poly(\max(\arb, \log n))$ for graphs of arbitrary arboricity. The subsequent algorithm of~\cite{MPC-MIS-loglogn} obtains the same degree reduction in $O(\log \log n)$ rounds. This result combined with the $O(\sqrt{\log \Delta}\log \log \Delta + \log \log n)$-round randomized \mis\ and \mm\ algorithm~\cite{doi:10.1137/1.9781611975482.99} for general graphs yields an $O(\sqrt{\log \arb}\log \log \arb + \log \log n)$ running time, i.e., a round complexity of $O(\log \log n)$ on graphs of polylogarithmic arboricity.

\textbf{Deterministic \mis\ and \mm} The fastest deterministic \mis\ and \mm\ algorithm for general graphs runs in $O(\log \Delta + \log \log n)$ by a result of~\cite{MPC-MIS-det-general}. When restricting our attention to deterministic algorithms in terms of the arboricity, only an $O(\log^2 \log n)$-time bound for \mis\ and \mm\ on trees (i.e., $\arb = 1$) is known due to Latypov and Uitto~\cite{MPC-MIS-det-log2logn}.

On the hardness side, there is a component-stable conditional lower bound of $\Omega(\log \log n)$ for \mm, which holds even on trees, by the works of \cite{CDP21lb, GKU19}.

\paragraph{Arboricity Coloring} %  with Linear Local Memory
The problem of coloring a graph is a corner-stone \textit{local} problem, with numerous applications, that has been studied extensively. A $K$-coloring of vertices is a function $\phi \colon V \rightarrow \{1,\ldots,K\}$ such that for every edge $\{u,v\} \in E$ vertices $u$ and $v$ are assigned distinct colors, i.e., $\phi(u) \neq \phi(v)$.
There has been a long succession of results based on recursive graph partitioning procedures in \congested\ and linear-memory \mpc~\cite{BKM20, CFG+19, HP15, P18, PS18}. The randomized complexity of $(\Delta+1)$-coloring and $(\Delta+1)$-list-coloring was settled to $O(1)$ by~\cite{CFG+19}, and \cite{CDP20cc1}~achieved the same result deterministically. 

Much of the research on this topic has focused on coloring with $\Delta + 1$ colors. However, for sparse graphs, the maximum degree $\Delta$ may be much larger than the arboricity $\arb$, and thus an arboricity-dependent coloring would be more satisfactory. 
Note that any graph admits a $2 \arb$-coloring, and that this bound is tight.
In this regard, Barenboim and Khazanov~\cite{BK18} observed that any problem can be solved in $O(\arb)$ deterministic time (since there are at most $\arb n$ edges), and proved a deterministic upper bound of $\arb^\eps$ to $O(\arb)$-color vertices, for any constant $\eps > 0$.  On the randomized side, Ghaffari and Sayyadi~\cite{GS19} show that an $O(\arb)$-coloring can be solved in $O(1)$ time, settling its complexity. 

\subsection{Deterministic Algorithms and Derandomization}\label{sec:rand-vs-det}
The derandomization of local algorithms has recently attracted much attention in the distributed setting (see, e.g., \cite{BKM20, coyconnectivity2021, CDP20cc1, MPC-MIS-det-general, CDP20ccb, FGG22P, PP22}). The clever derandomization scheme introduced by \cite{CPS17} enhances classic derandomization methods \cite{Lub93, MNN94} with the power of local computation and global communication, in spite of the limited bandwidth. Specifically, a randomized process that works under limited independence is derandomized, i.e., the right choice of (random) bits is found, by either brute-forcing all possible bit sequences or using the method of conditional expectation, that is, computing the conditional expectations in a distributed fashion.
There, one divides the seed into several parts and fixes one part at a time in a way that does not decrease the conditional expectation via global communication in a voting-like manner. 
In addition, to reduce the size of the search space, one can color the graph such that dependent entities are assigned distinct colors, i.e., their random decisions are decoupled, and entities that do not depend on each other can be safely assigned the same color. 

Interestingly, a good result is returned by the derandomized process even when concentration bounds fail to hold with high probability. On the other hand, the \textit{shattering effect}, that is, the emergence of (small) components of size $\poly(\Delta, \log n)$ after $O(\log \Delta)$ rounds, no longer appears under limited independence. 
Crucially, the lack of this phenomena poses major challenges in obtaining deterministic algorithms with 
complexities that match those of randomized algorithms. 
To overcome these challenges, current derandomization solutions often require large global memory (e.g., superlinear in the input size) or long local running time (e.g., large polynomial in $n$ or even exponential). 

\subsection{Our Contribution}
We introduce several improved deterministic algorithms for the three fundamental graph problems presented above in the \mpc\ setting. 

% MIS and MM
\paragraph{\mis\ and \mm}
We develop a deterministic low-space \mpc\ algorithm that reduces the maximum degree of the input graph to $\poly(\lambda)$. 

\begin{restatable}{theorem}{maindegred}\label{thm:main_deg_red}
    There is a $O(\log \log n)$-round deterministic low-space \mpc\ algorithm with linear global memory that computes an independent set $\mathcal{I}$ (or a matching $\mathcal{M}$) of a $\arb$-arboricity input graph $G$ such that the resulting graph $G \setminus \mathcal{I}$ (or $G \setminus \mathcal{M}$) has degree at most $\poly(\arb)$.
\end{restatable}

This algorithm (\Cref{chp:mismm}) is obtained by improving and derandomizing the $O(\log \log n)$-time algorithm of \cite{MPC-MIS-loglogn}, which achieves a reduction of $\poly(\max(\lambda, \log n))$. Our approach not only attains a stronger guarantee deterministically, but it also simplifies the intricate pipelined \mpc\ simulation~\cite{MPC-MIS-loglogn} of an arboricity-based \local\ degree reduction by~\cite{BEPS16}. Our end result is an $O(\log \arb + \log \log n)$-round algorithm by invoking the $O(\log \Delta + \log \log n)$ deterministic algorithm of~\cite{MPC-MIS-det-general} on the resulting $\poly(\arb)$-degree graph. This proves that \mis\ and \mm\ can be solved deterministically in $O(\log \log n)$ rounds on polylogarithmic arboricity graphs, implying up to an exponential speed up.
Prior to our work, a deterministic $O(\log \log n)$-round complexity was known only for $\poly(\log n)$-maximum degree graphs.

\begin{restatable}{theorem}{mainsuperlinear}\label{thm:main_superlinear}
    There is a deterministic low-space \mpc\ algorithm that computes a maximal independent set and a maximal matching in $O(\log \arb + \log \log n)$ rounds using $\localspace$ local space and $O\left(n^{1+\eps} + m \right)$ global space, for any constant $\eps > 0$.
\end{restatable}

When restricting our attention to the near-linear global memory regime (\Cref{sec:mis_near_global}), we obtain the first deterministic low-space \mis\ and \mm\ algorithms for general graphs. Each of the following round complexities is superior for different ranges of $\arb$.  
    \begin{restatable}{theorem}{mainlinear}\label{thm:mainlinear} 
        There is a \mpc\ algorithm that deterministically computes a \mis\ and a \mm\ with round complexity on the order of
        \[
            \min\left\{
                \begin{array}{l}
                    \log \arb \cdot \log\log_{\arb} n,\\
                    \arb^{1+\eps} + \log \log n \quad\text{for \mm},\\
                    \arb^{2+\eps} + \log \log n \quad\text{for \mis}
                    
                \end{array}
            \right\},
        \]
        for any constant $\eps > 0$, using $\localspace$ local space and $O\left(m + n \right)$ global space.
    \end{restatable}

The first runtime bound is obtained via a space-efficient \mpc\ derandomization of Luby's algorithm. In particular, we reduce the global memory usage of the $O(\log n)$-round algorithm of~\cite{MPC-MIS-det-general} from $O(n^{1+\eps}+m)$ to $O(n+m)$ by devising a pessimistic estimator that does not require knowledge of the two-hop neighborhood, i.e., it relies only on the one-hop neighborhoods to find a partial solution that decreases the problem size by a constant factor. Then, the graph exponentiation process yields the final runtime, where the $O(\log \arb)$-multiplicative factor stems from the $\arb^{\Omega(1)}$-progress needed at each iteration to respect the linear global memory constraint, as opposed to the  $O(\log \arb)$-additive term of \cref{thm:main_superlinear}.
For low-arboricity graphs, i.e., in the $(\log^{O(1)}\log n)$-range, we apply our $\poly(\arb)$-degree reduction followed by a fast reduction to arboricity coloring, and adapt coloring algorithms by Barenboim and Elkin~\cite{BE10} and Panconesi and Rizzi~\cite{PR01} to achieve the last two round complexities. Even for trees ($\arb = 1$) the best deterministic result had $O(\log^2 \log n)$ running time~\cite{MPC-MIS-det-log2logn}.

\paragraph{Arboricity Coloring}
We design a simple constant-round deterministic algorithm for the problem of computing a vertex coloring with $O(\arb)$ colors in the linear memory regime of \mpc, with relaxed global memory $n\cdot \poly(\arb)$.

\begin{restatable}{theorem}{mainarbcolalltoall}\label{thm:main_arbcol_alltoall}
        Given a graph $G$ with arboricity $\arb$, a legal $O(\arb)$ vertex coloring of $G$ can be computed deterministically in $O(1)$ on a linear-memory \mpc\ using $n \cdot \poly\arb$ global memory. 
\end{restatable}
    
Our constant-round algorithm is strikingly simple: After one single step of graph partitioning, all induced subgraphs are of linear size and can thus be collected and solved locally. We also extend this algorithm to the \congested\ model when $\poly(\lambda) = O(n^2)$. For the linear-memory \mpc\ setting with linear global memory, we show an $O(\arb)$-coloring algorithm running in $O(\log \arb)$ time. Both results (\Cref{chp:col}) improve exponentially upon the deterministic $(\arb^{\Omega(1)})$-round algorithm of \cite{BK18}.

\subsection{Comparison with the State-of-the-Art}
We review the main ideas behind our arboricity-dependent \mis, \mm, and coloring algorithms.

\paragraph{Arboricity-Dependent \mis\ and \mm\ Algorithms} 
The arboricity-dependent \mpc\ algorithms of~\cite{MPC-MIS-log2logn, MPC-MIS-loglogn} rely on a key degree reduction routine by~\cite{BEPS16}, which essentially reduces the number of nodes of degree at least $\beta$ by a $(\beta^{\Omega(1)})$-factor in $O(1)$ \local\ rounds for any $\beta \ge \poly(\log n, \arb)$ with high probability.

\subparagraph{\textit{Randomized Approaches.}} 
The $O(\log \log n \cdot \log \log \Delta)$-round algorithm of \cite{MPC-MIS-log2logn} shows that a black-box simulation of the above degree-reduction routine can be performed on a low-memory \mpc\ in $O(\log \log n)$ rounds to achieve a $(\Delta^{\eps})$-degree reduction, for any constant $\eps > 0$. Therefore, the repeated application of this routine for $O(\log \log \Delta)$ phases (one phase for each degree class with quadratically-spaced degree classes) leads to the claimed round complexity.
By using a pipelining idea, ~\cite{MPC-MIS-loglogn} shows that these $O(\log \log \Delta)$ phases can be simulated concurrently in $O(\log \log n)$ rounds. 
The idea being that as more vertices become irrelevant in the current phase, the next phase can be started concurrently with the current phase which has not necessarily finished completely.
However, in their framework, directly lifting multiple \local\ degree-reduction routines to low-space \mpc, which operate in conflicting parts of the same graph, results in a quite involved \mpc\ simulation and analysis, e.g., it requires phases to be asynchronous, a non-uniform graph exponentiation process, the concept of pending vertices, and keeping track of the message history of certain nodes.

\subparagraph{\textit{Our Deterministic Approach.}} 
Our deterministic approach is based on a novel pairwise independent analysis of the \local\ degree reduction by~\cite{BEPS16} and a clean \local\ pipelining framework, which makes derandomization in the \mpc\ model possible. We first inspect the internal details of the primitive by~\cite{BEPS16} to achieve the same expected degree reduction as under full independence solely relying on pairwise independence. We then define a \textit{pipelined} simulation of this routine by extracting multiple subgraphs, one for each degree class, on which we can make progress simultaneously. Thereby, we bring the \mpc\ pipeline framework of~\cite{MPC-MIS-loglogn} at the \local\ level. This locality makes the concurrency of the algorithm easier to handle and, thus, the \mpc\ simulation and its analysis significantly simpler. In fact, we show that our carefully-designed \local\ procedure can be simulated and derandomized on a low-space \mpc\ for multiple stages efficiently. 

We also show that initially reducing the size of all degree classes by a $\Delta^{\Omega(1)}$-factor makes it possible to carry out all phases synchronously with a uniform graph exponentiation process, i.e., at each step $k$ all nodes have knowledge of their $k$-hop neighborhood. Our algorithm overcomes several other challenges such as reducing the maximum degree of the input graph deterministically to $O(n^\delta)$, for any constant $\delta > 0$, and limiting the random bits for the derandomization of the $O(\log \log \Delta)$ algorithmic routines to $O(\log \Delta)$. On a qualitative level, our \local\ procedures based on shared random bits may provide a more intuitive way to convert \local\ randomized experiments to deterministic \mpc\ algorithms. 

\paragraph{Arboricity-Coloring Algorithms} 
Our approach builds upon the vertex partitioning techniques introduced by~\cite{CFG+19, CDP20cc1, GS19}. 
We derandomize a binning procedure that partitions the vertices of the input graph in $k$ subsets, each of which induces a subgraph of arboricity at most $O(\arb/k)$ so that the product of the number of bins and the induced arboricity is linear.
We then color each of the $k$ parts with $O(\arb/k)$ colors separately and independently. Moreover, bad-behaving nodes, which were excluded by the partitioning algorithm due to their high degree, will induce a subgraph of size $O(n)$ that can be colored locally using a fresh color palette of size $2\arb$. 

Our key technical contribution is a derandomization scheme that overcomes the inapplicability of the distributed method of conditional expectation. In fact, as opposed to the derandomization scheme for $(\Delta+1)$-coloring of~\cite{CDP20cc1}, arboricity is a global quality measure that cannot be decomposed into functions computable by gathering neighbors of each node on a single machine. To circumvent this, we find a subset of candidate partitions whose $k$ parts are of linear size and can be checked by a single machine, yet allowing only a $\poly(\arb)$ overhead factor in the global space.

We note that the randomized constant-round arboricity-coloring of \cite{GS19} uses a similar graph partitioning approach as a preprocessing step, which allows them to work with bins of arboricity $O(\log n)$ whenever $\arb = \Omega(\log n)$. However, the main challenge there is to solve the problem for $O(\log n)$-arboricity graphs. To do so, they take an orthogonal approach compared to ours, that is, a topology-gathering idea that simulates a $O(\log^2 \beta)$-round \congest\ coloring algorithm for $\beta = \poly(\log \log n)$-arboricity graphs in $O(1)$ time.

For the second algorithm, we observe that \cite{BK18} computes $O(\log \arb)$ layers of a graph decomposition known as $H$-partition, after which the graph has linear size. We show that applying the $(\Delta+1)$-list coloring algorithm of~\cite{CDP20cc1} on each layer yields an $O(\arb)$-coloring.

\section{Preliminaries}\label{chp:tools}
This section introduces the notation, useful definitions, and the algorithmic routines frequently used in the subsequent sections.

    \paragraph{Notation} For an integer $k \ge 1$, we use $[k]$ to denote the set $\{1,2,\dots,k\}$.
    Let~$G=(V,E)$ be the undirected $n$-node $m$-edge graph given in input with maximum degree $\Delta$.
    Denote $\deg_G(u)$ as the degree of vertex $u$ in graph $G$ and $N_G(u)$ as the set of neighbors of $u$ in $G$ (when $G$ is clear from the context, we may omit it).
    The distance from $u$ to $v$, $\dist_G(u,v)$, is the length of the shortest path from $u$ to $v$ in $G$. 
    For a vertex subset $U \subseteq V$ or edge subset $U\subseteq E$, we use $\deg_U(u)$ and $N_U(u)$ to refer to the subgraph of $G$ induced by $u \cup U$ denoted by $G[u \cup U]$. 
    For a vertex subset $V' \subseteq V$, let $N_G(V')$ be the union of the neighbors of each $v \in V'$.
    Let $\dist_G(u,V')$ be the minimum distance from $u$ to any vertex in $V'$.
    Further, we denote $G^k$, for $k \ge 2$, as the graph where each pair of vertices at distance at most $k$ in $G$ is connected. 
    The $k$-hop neighborhood of a vertex $v$ in $G$ is the subgraph containing all vertices $u$ satisfying $\dist_G(u,v)\le k$
    and the $k$-hop neighborhood of a vertex set $V'\subseteq V$ is the union of the $k$-hop neighbors of $v\in V'$. 
    Finally, denote $G \setminus \indepset$ as the graph obtained by removing the nodes in $\indepset$ and their neighbors. Similarly, $G \setminus \matching$ denotes the graph where nodes in $\matching$ are removed.

    \paragraph{Arboricity Definitions}
    The arboricity $\arb$ of $G$ is equal to the minimum number of forests into which the edge set of $G$ can be partitioned. It is closely linked to the following graph-theoretic structure known as $H$-partition~\cite{BE10}.
    \begin{definition}[$H$-Partition]\label{def:hpartition}
        An $H$-partition $\hpart$ of a $\arb$-arboricity graph $G$ parameterized by its \textit{degree} $d$, for any $d > 2\arb$, is a partition of the vertices of $G$ into layers $H_1, \ldots, H_\last$ such that every vertex $v \in H_i$ has at most $d$ neighbors in layers with indexes equal to or higher than $i$, i.e., $\bigcup_{j=i}^{\last} L_j$.
    \end{definition}
    \fullOnly{
    We defer arboricity properties and the algorithmic construction of an $H$-partition to \Cref{sec:graph_arboricity}.}
    \shortOnly{We defer arboricity properties and the algorithmic construction of an $H$-partition to the full version of the paper.}

    \paragraph{Primitives in Low-Space \mpc}
    We use the \mpc\ primitives of the following lemma as black-box tools, which allow us to perform basic computations on graphs deterministically in $O(1)$ rounds. These include the tasks of computing the degree of every vertex, ensuring neighborhoods of all vertices are stored on contiguous machines, collecting the $2$-hop neighborhoods, etc.
    \begin{lemma}[\cite{G99, GSZ11}]\label{lem:primitives}
        For any positive constant $\spaceconst$, sorting, filtering, prefix sum, predecessor, duplicate removal, and colored summation task on a sequence of $n$ tuples can be performed deterministically in MapReduce (and therefore in the \mpc\ model) in a constant number of rounds using $\memory = n^\spaceconst$ space per machine, $O(n)$ global space, and $\widetilde{O}(n)$ total computation.
    \end{lemma} 

    \paragraph{Derandomization Framework}
    \fullOnly{
    The derandomization techniques applied in all-to-all communication models are described in more detail in \Cref{sec:derandomization_framework}.}
    \shortOnly{The derandomization techniques applied in all-to-all communication models are described in more detail in the full version of the paper.} We here explain a simple coloring trick that has been frequently used to reduce the number of random bits needed.
    Specifically, if the outcome of a vertex depends only on the random choices of its neighbors at distance at most $t$, then $k$-wise independence among random variables of vertices within distance $t$ is sufficient. Whenever this is the case, we find a coloring that maps vertex IDs to shorter names, i.e., colors, such that ``adjacent'' vertices are assigned different names. To achieve that, we apply Linial's coloring algorithm~\cite{doi:10.1137/0221015} as stated in the following lemma and based on an explicit algebraic construction~\cite{EFF85, 10.1145/1583991.1584032}. For a proof of its \mpc\ adaptation, we refer to~\cite{FGG22P}.
    \begin{lemma} \label{lm:coloring}
        Let $G = (V,E)$ be a graph of maximum degree $\Delta \leq n^{\spaceconst}$ and assume we are given a proper $C$-coloring of $G$. Then, there is a deterministic algorithm that computes a $O(\Delta^2 \log^2 C)$ coloring of $G$ in $O(1)$ \mpc\ rounds using $O(n^\spaceconst)$ local space, $O(n + m)$ global space, and $\widetilde{O}(n \cdot \poly(\Delta))$ total computation. If $C = O(\Delta^3)$, then the number of colors can be reduced to $O(\Delta^2)$.
    \end{lemma}

    \fullOnly{    
        \paragraph{Graph Exponentiation}
        A central technique to simulate \local\ algorithms exponentially faster via all-to-all communication is that of letting each node send its $r$-hop neighborhood to each of the nodes in it, so that in one round every node learns the topology of the ball of radius $2r$ around itself~\cite{LW10}. Then, each node can simulate $2r$ rounds by computing the choices made by its $2r$-hop neighbors. The repeated application of this process often leads to an exponential speed up provided that the exchanged information respects bandwidth and both local and global space constraints. Our \mis\ and \mm\ algorithms are based on this technique, which will be explained case by case.%
    }
\section{Deterministic Arboricity-Based Degree Reduction}\label{chp:mismm} 
Our degree-reduction procedure operates on multiple degree classes simultaneously. \Cref{subsec:local_alg} is devoted to our degree-reduction algorithm in the \local\ model, as overviewed next. We first consider a single degree class and then switch our view to multiple classes. Our single-class degree reduction routine consists of two phases. The first phase (\Cref{lem:degree_reduction_extract_single}) is a preparation step where a subgraph with a certain structure is computed. In the second phase, we find a partial solution on this subgraph using \Cref{lem:degree_reduction_pairwise}. The multi-class degree reduction routine has two analogous phases. In \Cref{lem:degree_reduction_extract_all}, we extend the single-class preparation phase to work with multiple degree classes concurrently. We do so by computing a suitable subgraph for each degree class such that subgraphs from different classes are non-conflicting. Then, by finding partial solutions for each degree class on the subgraphs from the previous step, we show that we can make a global progress without requiring additional random bits (\Cref{lem:degree_reduction_multi_single}). In addition, we consider multiple consecutive runs of our multi-class degree reduction routine in \Cref{lem:degree_reduction_multi_multi}, which will be crucial to achieve the \mpc\ speed up via graph exponentiation. 
Finally, in \Cref{sec:mpc_part}, we proceed to derandomize and simulate (multiple) \local\ degree reduction phases efficiently on an \mpc\ with strongly sublinear memory.

\subsection{The \local\ Algorithm}\label{subsec:local_alg}
It turns out that so-called $\beta$-high graphs play an important role in finding partial \mis\ and \mm\ solutions.

    \begin{definition}[$\degreehigh$-high-graph]
    Let $\degreehigh$ be a parameter.
    Let $H$ be a graph on vertices $V(H) = V^{high} \sqcup V^{low}$ where $\sqcup$ denotes $V^{high} \cap V^{low} = \emptyset$. 
    We say that $H$ is a $\degreehigh$-high graph if the following holds:
        \begin{enumerate}
            \item $\deg_{V^{low}}(v) \ge \degreehigh^{4}$ for every $v \in V^{high}$, and \\
            \item $\deg_{V^{high}}(v) \le \degreehigh$ and $\deg_{V^{low}}(v) \le \degreehigh^2$ for every $v \in V^{low}$.
        \end{enumerate}
    \end{definition}

    We use a $\degreehigh$-high graph to find an \mis\ (or a \mm) on low-degree nodes, using a random process based on pairwise independence, such that only a $\degreehigh$-fraction of high-degree nodes remain in expectation, as the following lemma shows.

\begin{lemma}\label{lem:degree_reduction_pairwise}
    Let $H$ be a $\degreehigh$-high graph.
    There is a \local\ algorithm running in $O(1)$ rounds that computes an independent set $\indepset$ (or a matching $\matching$) of $H$ such that in the remaining graph $H' = H \setminus \indepset$ (or $H \setminus\matching$) the following holds
    \[\E[|V^{high}(H')|] \leq \frac{\card{V^{high}(H)}}{\degreehigh^{\Omega(1)}}.\]
    Moreover, the only source of randomness in the algorithm is $O(\log C)$ shared random bits with $\poly(\Delta) \le C \le \poly(n)$, where $C$ is an upper bound on the number of colors (or IDs) that we assume to be assigned to vertices in $V^{low}$ such that any two nodes whose random choices are assumed to be independent are given different colors.
\end{lemma}
\begin{proof}
    We analyze a random sampling approach that builds an independent set $\indepset$ on the vertex set $V^{low}$  (or matching $\matching \subseteq E(H)$).
    For \mis, we sample each node $u\in V^{low}$ with probability $p = 1/\pvalue$ using pairwise independent random variables. Then, $u$ joins $\indepset$ if none of its neighbors is sampled. Consider now an arbitrary node $v$ in $V^{high}$. Let $\{ X_u \}_{u \in N_{V^{low}}(v)}$ be the random variables denoting the event that $u$ joins $\indepset$ and denote $X = \sum_{u \in N_{V^{low}}(v)} X_u$ as their sum. We have
    \begin{align*}\label{eq:arbmisfirstmoment}
        \Pr[X_u = 1] &\ge \Pr\left[u \text{ sampled}\right] - \sum_{\substack{u' \in N_{V^{low}}(u)}}\Pr[u,\, u' \text{ sampled}]  \ge p - \beta^2p^2,
    \end{align*} 
    where the second inequality follows from pairwise independence. It follows that the expected number of neighbors of $v$ in $\indepset$ is at least $\E[X] = \sum_{u \in N_{V^{low}}(v)} \Pr[X_u = 1] \ge \beta^4\cdot \left(p - \beta^2p^2\right) \, \eqdef \,  \mu.$
    Our goal is now to prove that $\Pr\left[X = 0\right] \le \frac{\Var[X]}{\mu^2} \le \frac{5}{\beta}$ by applying Chebyshev's inequality \fullOnly{(\Cref{thm:cheb})}.
    For any two vertices $u, u' \in N(v)$, we have that $\E[X_u X_{u'}]  \le p^2$ by pairwise independence.
    Then,
    \begin{align*}
        \Var[X_u] & = \E[X_u^2] - \E[X_u]^2 \le p - p^2(1 - \beta^2 p)^2 \le p,\\
        \Cov[X_u, X_{u'}] & = \E[X_u X_{u'}] - \E[X_u]\E[X_{u'}] \le p^2 - p^2(1 - \beta^2 p)^2 \le  2\beta^2 p^3,
    \end{align*}
    where we use that $p \le \frac{1}{2\beta^2}$.
    Then, by observing that $\mu^2 = \beta^8 \cdot \left(p - \beta^2p^2\right)^2 \ge \frac 1 2 \beta^8 \cdot p^2$ since $p \le \frac{1}{4\beta^2}$, we get
    \begin{align*}
        \frac{\Var[X]}{\mu^2} &\le \frac{\sum\limits_{w \in N_{V^{low}}(v)} \Var[X_w] + \sum\limits_{w, w' \in N_{V^{low}}(v)} \Cov[X_w, X_{w'}]}{\mu^2} \\ &\le \frac{2}{\beta^4 p} + 4\beta^2p  
        \le \frac{5}{\beta},
    \end{align*}
    where the last inequality follows from our choice of $p=1/\pvalue$. Therefore, the expected number of $V^{high}$-nodes that do not have a neighbor in $\indepset$ is at most $5 |V^{high}| / \beta$, as desired.
    For the case of \mm, we let each $u \in V^{low}$ propose that a randomly chosen edge $\{u,v\}$ with $v \in V^{high}$ is added to the matching. Any $v \in V^{high}$ receiving a proposal accepts one arbitrarily. 
    A node $v \in V^{high}$ has at least $\beta^4$ neighbors $u \in V^{low}$ with $\deg_{V^{low}}(u) < \beta$. Thus, the number of proposals that $v$ receives in expectation is at least $\beta^3$. By Chebyshev's inequality for pairwise independent random variables \fullOnly{(\Cref{cor:chebpw})}, $v$ is unmatched with probability at most $\frac{5}{\beta}$.

    Next, we prove the bound on the amount of randomness required by the above process. Let us assume that nodes in $V^{low}$ are assigned colors (or IDs), $\phi \colon V \rightarrow [C]$, such that nodes whose (random) choices are assumed to be pairwise independent are given different colors.
    \fullOnly{
    The set of colors is mapped to $\{0,1\}^k$, for some integer $k > 0$, using a $2$-wise independent hash family $\mathcal{H}$ of size $O(\log C + \log \degreehigh) = O(\log C)$ using \Cref{lemma-hash}.}
    \shortOnly{The set of colors is mapped to $\{0,1\}^k$, for some integer $k > 0$, using a $2$-wise independent hash family $\mathcal{H}$ of size $O(\log C + \log \degreehigh) = O(\log C)$.}
    For \mis, let  $\mathcal{H} = \{h \colon [C] \mapsto \{0,1\}^{3\log \degreehigh}\}$. Every $h \in \mathcal{H}$ induces an independent set $\mathcal{I}_h$. If a vertex $v \in V^{low}$ with $h(\phi(v)) = 0$ has no neighbors $w \in V^{low}$ with $h(\phi(w)) = 0$, then $v$ is in $\mathcal{I}_h$. Similarly, for \mm, let $\mathcal{H} = \{h \colon [C] \mapsto \{0,1\}^{\log \degreehigh}\}$. Each $v \in V^{low}$ selects an arbitrary neighboring vertex $u \in V^{high}$ with $h(\phi(v)) = u$ by mapping in an arbitrary but consistent manner its neighbors to $\{0, \ldots, \degreehigh - 1\}$. Then, $\mathcal{M}_{h}$ includes each $V^{high}$-node with at least one proposal. By the above analysis, there is one hash function $h^* \in \mathcal{H}$ that achieves the expected result. 
\end{proof}
    Next, we introduce a preparation phase used to construct a $\degreehigh$-high-subgraph of a graph $G$ by appointing a subset of high-degree nodes with many low-degree neighbors. 
\begin{lemma}[Single-Class Preparation \cite{BEPS16}]\label{lem:degree_reduction_extract_single}
    Let $\Delta_i \geq \poly(\arb)$ be a parameter and $V(G, \Delta_i) :=\{v: \deg_G(v) > \highdeg\}$ denote the set of high-degree vertices in $G$. 
    There is a deterministic constant-round \local\ algorithm that finds a $\beta$-high-subgraph $H$ of $G$ with $\beta = \betadeg_i$, such that:
    \begin{enumerate}
        \item $V^{high}(H) \subseteq V(G,\Delta_{i})$.\label{local_single_A}
        \item $\dist(u, V(G,\Delta_{i})) \le 3$ for every $u \in V(H)$, i.e., no vertex outside the $3$-hop neighborhoods of $V(G,\Delta_{i})$ is in $H.$\label{local_single_B}
        \item $|V(G,\Delta_{i}) \setminus V^{high}(H)| \leq \frac{|V(G,\Delta_{i})|}{\Delta^{\Omega(1)}_i}$.\label{local_single_C}
    \end{enumerate}
\end{lemma}
    \fullOnly{The proof can be found in \Cref{sec:missing_proofs}. }
    \shortOnly{The proof is deferred to the full version of the paper.}
    We now extend the above preparation phase to work on multiple degree classes concurrently by \textit{clustering} nodes according to their highest-degree vertex at distance at most $4$.
    Let $i_{min} = \floor{\log \log (\arb^{\Theta(1)})}$ and $i_{max} = \ceil{\log \log \Delta}$ be two parameters. We define $\Delta_i = 2^{2^i}$ for $i \in [i_{min}, i_{max}]$ and the degree classes $V_i(G) = \{v \in V(G) \colon \deg_G(v) \in (\Delta_{i-1},\Delta_i]\}$ for a given graph $G$. We also define the sets $V_{\geq i}(G) = \bigcup_{j \geq i} V_j(G)$ and $V_{> i}(G) = \bigcup_{j > i} V_j(G)$. 
\begin{lemma}[Multi-Class Preparation]\label{lem:degree_reduction_extract_all}
    There is a deterministic \local\ algorithm running in $O(1)$ rounds that constructs a $(\Delta_i^{\Omega(1)})$-high-subgraph $H_i$ of a given graph $G$ for every $i \geq i_{min}$ such that:
    \begin{enumerate}
        \item $V^{high}(H_i) \subseteq V_i(G)$.\label{local_all_A}
        \item $|V_{\ge i}(G) \setminus V^{high}(H_i)| \leq |V_i(G)|/\Delta^{\Omega(1)}_i + \sum_{j > i} \Delta^{O(1)}_j |V_j(G)|$.\label{local_all_B}
        \item $\dist_G(V(H_j),V(H_i)) \geq 2$ for every $j \geq i_{min}$ with $j \neq i$.\label{local_all_C}
    \end{enumerate}
\end{lemma}
\begin{proof}
    Let $D_G(v) = \displaystyle\max_{\dist_G(u,v) \le 4} \deg_G(u)$ for every $v \in V(G)$. Then, denote \[G_i = G[\{u \in V(G) \mid D_G(u) \in (\Delta_{i-1},\Delta_i] \text{ and } \dist(u, V_i(G)) \le 3\}]\] as the subgraph of $G$ induced by vertices whose highest-degree vertex in their $3$-hop neighborhood is in $V_i(G)$, for every $i$. 
    Consider now any $j > i$. It holds that $\dist_G(V(G_j),V(G_i)) \geq 2$. To see why, consider any $u \in V(G_i)$ and $v \in V(G_j)$. We have
    \begin{align*}5 &\le dist_G(u,V_{>i}(G)) \leq dist_G(u,V_j(G)) \leq  dist_G(u,v) + dist(v,V_j(G)) \leq dist_G(u,v) + 3\end{align*}
    and therefore $dist_G(u,v) \geq 2$, proving Property~\ref{local_all_C}.
    Also, note that $V(G_i) \cap V_{>i}(G) = \emptyset$ and therefore $\Delta(G_i) \leq \Delta_i$. 
    Now, let $H_i$ be the subgraph of $G_i$ that we obtain by invoking \cref{lem:degree_reduction_extract_single} with input $G_i$ and $\Delta_i$. 
    The graph $H_i$ is a $(\Delta^{\Omega(1)}_i)$-high-graph. 
    Observe that $V_i(G_i) = \{v \in V(G_i) \mid \deg_{G_i}(v) > \sqrt{\Delta_i}\} =: V(G_i,\Delta_i)$. By Property~\ref{local_single_A} of~\cref{lem:degree_reduction_extract_single}, $V^{high}(H_i) \subseteq V_i(G_i)$, and together with $V_i(G_i) \subseteq V_i(G)$, this implies $V^{high}(H_i) \subseteq V_i(G)$, which proves Property \ref{local_all_A} of this lemma. Moreover, Property~\ref{local_single_C} of~\cref{lem:degree_reduction_extract_single} together with the fact that $i_{min} =\Theta(\log \log \arb)$ gives
    \[|V_i(G_i) \setminus V^{high}(H_i)| \leq  \frac{|V_i(G_i)|}{\Delta_i^{\Omega(1)}} \leq \frac{|V_i(G)|}{\Delta_i^{\Omega(1)}}.\]
    Then, consider the remaining vertices $v \in V_i(G) \setminus V_i(G_i)$. Since $\dist_G(v,V_i(G)) = 0$, it holds that $D_G(v) > \Delta_i$. Thus,
    \begin{align*}|V_{\ge i}(G) \setminus V_i(G_i)|& \leq |\{v \in V(G) \mid \dist_G(v,V_{>i}(G)) \le 4\}| \leq \sum_{j > i}\Delta^{O(1)}_j|V_j(G)|,\end{align*}
    which combined with the above bound completes the proof of Property~\ref{local_all_B}. Finally, as $G_i$'s are node-disjoint, we can compute all $H_i$'s simultaneously in $O(1)$ \local\ rounds.
\end{proof}
Next, we show how to use the subgraphs constructed above to run a single stage of our algorithmic routine on all degree classes simultaneously.
\begin{lemma}[Multi-Class Single-Stage Algorithm]\label{lem:degree_reduction_multi_single}
    There is a \local\ algorithm that computes an independent set $\indepset$ (or a matching $\matching$) of a given graph $G$ such that in the remaining graph $G' = G \setminus \indepset$ (or $G \setminus\matching$) with strictly positive probability the following holds: 
    \[|V_{\ge i}(G')| \leq \frac{|V_{i}(G)|}{\Delta_i^{\Omega(1)}} + \sum_{j > i} \Delta_j^{O(1)} |V_j(G)|,\]
    for every $i \in [i_{min}, i_{max}]$. Moreover, the only source of randomness in the algorithm is $O(\log C)$ shared random bits with $\poly(\Delta) \le C \le \poly(n)$, where $C$ is an upper bound on the number of colors (or IDs) that we assume to be assigned to vertices such that any two nodes at distance at most $4$ in $G$ are given different colors.
\end{lemma}
\begin{proof}
    Let us analyze the independent set case, the argument for the matching case is essentially the same.
    First, we compute in $O(1)$ rounds, for every $i \in [i_{min}, i_{max}]$, a subgraph $H_i$ of $G$ using the algorithm of \cref{lem:degree_reduction_extract_all}.
    Then, we simultaneously run the algorithm of \cref{lem:degree_reduction_pairwise} on each $H_i$ to compute an independent set $\indepset_i$, which is independent also in $G$ since $H_i$ is obtained from $G$ without discarding any edge, i.e., $H_i = G[V(H_i)]$. Let $\indepset = \bigcup_{i} \indepset_i$, which is valid since any two subgraphs $H_i, H_j$ are at distance at least two by Property~\ref{local_all_C} of \Cref{lem:degree_reduction_extract_all}.
    Moreover, by Property~\ref{local_all_A}~and~\ref{local_all_B} of \cref{lem:degree_reduction_pairwise}, it is enough to show that $\card{V^{high}(H_i) \setminus (\indepset_i \cup N_G(\indepset_i))} \le \card{V^{high}(H_i)}/\Delta_i^{\Omega(1)}$ for each $i$. By an application of Markov's inequality, \cref{lem:degree_reduction_pairwise} fails in achieving the above guarantee for a fixed $i$ with probability at most $\Delta^{-\Omega(1)}_i$. By a union bound, the total failure probability is upper bounded by $\sum_{i} \frac{1}{\Delta_i^{\Omega(1)}} \le 1/2$ since we can assume that $\Delta_{i_{min}} = \poly(\arb) = \Omega(1)$. Therefore, all the at most $i_{max}$ runs of \cref{lem:degree_reduction_pairwise} succeed with strictly positive probability. This means that there is a sequence of $O(\log C)$ random bits that makes all subgraphs $H_i \setminus \indepset_i$ good. Note also that nodes whose choices are pairwise independent in \cref{lem:degree_reduction_pairwise} are given distinct colors (or IDs) by the assumption in the statement of this lemma.
\end{proof}
    Running the above algorithm for multiple stages in a sequential fashion leads to the following lemma. \fullOnly{Its proof can be found in \Cref{sec:missing_proofs}}\shortOnly{Its proof can be found in the full version of the paper.}. Intuitively, to achieve the stated upper bound we show that higher degree classes affect lower degree classes only asymptotically, since the sum of the sizes of our quadratically-spaced degree classes turns out to be dominated by a geometric series. This fact is also used in \cite{MPC-MIS-loglogn}.
\begin{lemma}[Multi-Class Multi-Stage Algorithm]\label{lem:degree_reduction_multi_multi}
    Let $k \in \mathbb{N}$ and assume that $k \ge s$ for an absolute constant $s$. There is a \local\ algorithm running in $O(k)$ rounds that computes an independent set $\indepset$ (or a matching $\matching$) of a given graph $G$ such that in the remaining graph $G' = G \setminus \indepset$ (or $G \setminus\matching$) with strictly positive probability the following holds: 
    \[|V_{\ge i}(G')| \leq |V_{\ge i}(G)|/\Delta_i^{\Omega(k)},\]
    for each $i \in [i_{min}, i_{max}]$. Moreover, $O(k \cdot \log C)$ shared random bits are the only source of randomness in the algorithm with $C$ defined as above.
\end{lemma}
\subsection{\mpc\ Simulation and Derandomization}\label{sec:mpc_part}
    Let us now turn our attention to the low-memory \mpc\ model. In the following, we assume that $\Delta = O(n^\delta)$, for any $\delta > 0$, and explain how to lift this assumption last.
    First, we show that the derandomization of one stage of our \local\ single-class degree reduction routine from \cref{lem:degree_reduction_pairwise} and \cref{lem:degree_reduction_extract_single} leads to the following result
    . On a high-level, we translate \cref{lem:degree_reduction_extract_single} to the \mpc\ setting and efficiently derandomize the pairwise independent random process of \cref{lem:degree_reduction_pairwise}, using the power of all-to-all communication and the coloring trick. \fullOnly{Its proof is deferred to \Cref{sec:missing_proofs}.}\shortOnly{Its proof is deferred to the full version of the paper.}
    \begin{lemma}[Single-Class Reduction]\label{lem:degree_reduction_mpc_single_round}
        For any fixed $i \in [i_{min}, i_{max}]$, there is a deterministic low-space \mpc\ algorithm running in $O(c)$ rounds that computes an independent set $\indepset$ (or a matching $\matching$) of a given graph $G$ such that in the remaining graph $G' = G \setminus \indepset$ (or $G \setminus\matching$) the following holds:
            \[|V_{\ge i}(G')| \leq |V_{\ge i}(G)|/\Delta^{c}_{i}.\]
    \end{lemma}
    \fullOnly{
    The above deterministic single-class reduction is applied to each degree class sequentially in the next lemma, whose proof can be found in \Cref{sec:missing_proofs}. It will be used as a preprocessing step in our final degree-reduction algorithm.}
    \shortOnly{The above deterministic single-class reduction is applied to each degree class sequentially in the next lemma, whose proof can be found in the full version of the paper. It will be used as a preprocessing step in our final degree-reduction algorithm}
\begin{lemma}[Preprocessing \mpc]\label{lem:degree_reduction_mpc_preprocessing}
    Let $c \in \mathbb{N}$ be a parameter. There exists a deterministic low-memory \mpc\ algorithm running in $O(c \log \log \Delta)$ rounds that computes an independent set $\indepset$ (or a matching $\matching$) of a given graph $G$ such that in the remaining graph $G' = G \setminus \indepset$ (or $G \setminus\matching$) the following holds for every $i \in [i_{min}, i_{max}]$:
        \[|V_{\ge i}(G')| \leq |V_{\ge i}(G)|/\Delta^{c}_{i}.\]
\end{lemma}
    The next lemma is a key technical step to achieve our end result. It shows how to efficiently derandomize our multi-class multi-stage \local\ algorithm using knowledge of the $k$-hop neighborhood for $k$ large enough. In particular, this knowledge will be used to compute a coloring of $G^4$ and to derandomize the \local\ algorithm \cref{lem:degree_reduction_multi_multi} running in $O(k)$ \local\ rounds in a constant number of \mpc\ rounds. This derandomization is performed by brute-forcing all possible bit sequences on each single machine. \fullOnly{Its proof is deferred to \Cref{sec:missing_proofs}.}\shortOnly{Its proof is deferred to the full version of the paper.}
\begin{lemma}\label{lem:mpc_degree_reduction_multi_multi}
    Let $k \in \mathbb{N}$. There is a deterministic low-memory \mpc\ algorithm that in $O(1)$ rounds computes an independent set $\indepset$ (or a matching $\matching$) of a given graph $G$ with maximum degree $\Delta_{sup}$ such that in the remaining graph $G' = G \setminus \indepset$ (or $G \setminus\matching$) the following holds for every $i \in [i_{min}, i_{max}]$
    \[|V_{\ge i}(G')| \leq |V_{\ge i}(G)|/\Delta_i^{\Omega(k)},\]
    provided that the following conditions are satisfied:
    \begin{enumerate}
        \item Each node $v \in V(G)$ has stored its $k$-hop neighborhood in $G$ on a single machine.
        \item $k \ge c'$ for some constant $c'$ and $k \leq \frac{\log(n)}{\log(\Delta_{sup} \cdot \log^{(k)} n))}$.
    \end{enumerate}
\end{lemma}
We are now ready to present our deterministic low-memory \mpc\ $\poly(\arb)$-degree reduction.
\maindegred*
    \begin{proof}[Proof of \Cref{thm:main_deg_red}]
        The first step of the algorithm is to apply the preprocessing of \Cref{lem:degree_reduction_mpc_preprocessing}, which takes $O(\log \log \Delta)$ rounds. Then, for $\ell = 0,1,\ldots,O(\log \log n)$ constant-round iterations, we perform the following steps: (i) run our \mpc\ multi-class multi-stage algorithm (\Cref{lem:mpc_degree_reduction_multi_multi}) for $c = O(1)$ times; and (ii) perform one step of graph exponentiation.
        Then, the returned independent set (or matching) is given by the union of the partial solutions computed in each iteration at step (i). 
        To execute step (i), any two consecutive runs of \Cref{lem:mpc_degree_reduction_multi_multi} are interleaved by a constant-round coordination step where we update the $2^\ell$-hop neighborhood of each (survived) node as follows. Suppose that $u \in V$ knows its $2^\ell$-hop neighborhood at the beginning of the first run in graph $G^{(\ell,1)}$. Notice that $G^{(\ell,2)}$ is obtained from $G^{(\ell,1)}$ by removing $\indepset_{h^*} \cup N(\indepset_{h^*})$. Therefore, to update the $2^\ell$-hop neighborhood of $u$ in $G^{(\ell,2)}$, each node $v \in \indepset_{h^*} \cup N(\indepset_{h^*})$ communicates that it is part of the current solution to all nodes in its $2^\ell$-hop neighborhood in $G^{(\ell,1)}$. This process is then repeated another $c-1$ times (the matching case is equivalent).
        
        The crux of the analysis is to prove that $V_{\ge i}(G^{(\ell_i)}) = \emptyset$ by the end of iteration $\ell_i = \floor{\log \alpha \log_{\Delta_i} n} - 1$. Recall that by the preprocessing step we can assume that $V_{\ge i}(G^{(0)}) \le V_{\ge i}(G)/\Delta_i^t$ for any constant $t > 0$. This allows us to conclude that $V_{\ge i}(G^{(\ell_i)}) = \emptyset$ for any $\ell_i = O(1)$. For the other values of $\ell_i$, the proof is by induction on $i$. Observe that $\ell_i < \ell_j$ for any $j < i \in [i_{min}, i_{max}]$ and let $k := 2^{\ell_i} = (\alpha \log_{\Delta_i} n) \ge c'$ since $\ell_i = \Omega(1)$. Fix an arbitrary $i$. After iteration $\ell_i$, the algorithm has performed $\ell_i$ graph exponentiation steps, implying that each node knows its $k$-hop neighborhood. One can verify that also $k \leq \frac{\log(n)}{\log(\Delta(G^{(\ell)}) \cdot \log^{(k)} n))}$ holds. Therefore, by $c$ repeated invocations of \Cref{lem:mpc_degree_reduction_multi_multi} in step (i), we get $|V_{\ge i}(G^{\ell_i+1})| \leq |V_{\ge i}(G^{\ell_i})|/\Delta_i^{c \cdot \Omega(k)} = 0$ for $c$ suff. large, since $k = \Omega(\alpha \log_{\Delta_i} n)$. 
        
        This proves that after $\ell_{i_{min}} = O(\log \log n)$ rounds our degree reduction is completed. Moreover, the local memory constraint is also respected. To see why, observe that at any iteration $\ell \in (\ell_{i+1}, \ell_{i}]$, the size of the $k$-hop neighborhood is at most $\Delta_i^{k+1} = \localspace$ because all higher degree classes are empty by the end of the $(\ell_{i+1})$-th iteration, i.e., $\Delta_i$ is the maximum degree in $G^{(\ell_{i+1}+1)}, \ldots, G^{(\ell_{i})}$.
        Finally, we prove that the global memory usage is linear in the input size. Let $D_\ell(v,d) = \displaystyle\max_{\dist(u,v) \le d, \text{ in }G^{(\ell)}} \deg_{G^{(\ell}}(u)$. The total memory required for the graph exponentiation process at iteration $\ell$ is upper bounded by $\sum_{v \in V(G^{\ell})} D_\ell(v,k)^k \le \sum_{i} V_i(G^{\ell}) \Delta_i^{2k+1} $ . This inequality can be seen as a charging scheme where $v$ charges its highest-degree $k$-hop neighbor, which has degree at most $\Delta_i$ and thus $v$'s neighborhood requires at most $\Delta_i^{k+1}$ words for being stored. Moreover, the total number of nodes that charge $v \in V_i$ is at most $\Delta_i^k$. By our preprocessing step and step (i) of the algorithm, it holds that $V_{\ge i}(G^{\ell+1}) \le  V_{\ge i}(G)/\Delta_i^{2k+1}$. Observe that for small values of $k = O(1)$, we may not be able to run \Cref{lem:mpc_degree_reduction_multi_multi} due to the condition $k \ge c'$ but we can choose $t$ such that $V_{\ge i}(G^{\ell}) \le V_{\ge i}(G)/\Delta_i^{t} \le V_{\ge i}(G)/\Delta_i^{2k+2}$ for $k = O(1)$. Thus,
        \begin{align*}
            \sum_{i} V_i(G^{\ell}) \Delta_i^{2k+1} \le \sum_{i} \frac{V_{\ge i}(G)}{\Delta_i} = O(m).
        \end{align*}
        This proves that the global memory constraint is respected for vertices in the $k$-hop neighborhood of $V_{\ge i_{min}}(G)$. A vertex with degree at most $\poly(\arb)$ whose highest degree vertex in its $8$-hop neighborhood is at most $\poly(\arb)$ will remain in $G$ and its degree will not change throughout the execution of our algorithm, thus we can safely \textit{freeze} it, i.e., it will not run the algorithm and thus the overall global space constraint is also respected.
    \end{proof}
Finally, we combine the above degree reduction with the $O(\log \Delta + \log \log n)$-round algorithm of~\cite{MPC-MIS-det-general} to prove \cref{thm:main_superlinear}. \fullOnly{We include its proof together with the explanation of how to lift the assumption on the maximum degree in \Cref{sec:missing_proofs}.}
\shortOnly{The proof together with the explanation of how to lift the assumption on the maximum degree can be found in the full version of the paper.}
\section{\mis\ and \mm\ with Near-Linear Global Memory}\label{sec:mis_near_global}
    In the sublinear regime of \mpc\ with near-linear global memory, i.e., $\tilde O(n+m)$, we obtain the following three running time bounds, each of which is asymptotically superior for different values of $\arb$.
    \mainlinear*

    \subsection{High-Arboricity Case}\label{sec:higharbcase}
    In this section, we show a $O(\log n)$ runtime for high-arboricity graphs, which will be helpful to prove the $O(\log \arb \log \log_{\arb} n)$ result. When the input graph $G$ has arboricity $\arb = n^{\Omega(1)}$, any $O(\log \arb)$-round algorithm can spend $O(\log n)$ rounds for computing a solution. 
    The main challenge in converting classical $O(\log n)$-round Luby's algorithms to low-space \mpc\ algorithms is that a machine may not be able to store the whole neighborhood of a node in its memory. To overcome such limitation, Czumaj, Davies, and Parter~\cite{MPC-MIS-det-general} devise a constant-round sparsification procedure that reduces the maximum degree to $O(n^\delta)$, for any constant $\delta>0$, while preserving a partial solution that decreases the problem size by a constant factor. This allows to convert each randomized step of Luby's into a deterministic step. 
    Crucially, the derandomization of each phase of Luby's requires storing the two-hop neighborhoods on a single machine to check whether a node (or edge) survives to the next phase. This results in a $\Omega(n^{1+2\delta})$-global space usage, as each vertex has up to $O(n^{2\delta})$ many 2-hop neighbors. 
    
    Achieving $O(n+m)$-global space would require computing a partial solution that produces asymptotically the same results and can be computed using only knowledge of the one-hop neighborhood. To this end, we introduce a \textit{pessimistic estimator} that allows to derandomize a variant of Luby's algorithm in the same asymptotic number of rounds and using linear global memory. 
    \fullOnly{
    In the following, we provide a high-level description of the algorithm in the \mis\ case and refer to \Cref{sec:higharbmis} for its proof covering both cases.}
    \shortOnly{In the following, we provide a high-level description of the algorithm in the \mis\ case and refer to the full version for its proof covering both cases.}
    
    We consider a single iteration of the algorithm, where the objective is to find a partial solution that decreases the problem size by a constant factor. Let \textit{good} nodes be a subset of nodes such that: (i) each good node has constant probability of being removed in this iteration, and (ii) good nodes are incident to a constant fraction of edges in the remaining graph. 
    
    The first step is to apply the sparsification procedure of~\cite{MPC-MIS-det-general}, which works with $O(n+m)$ total space. It returns a subgraph that includes good nodes and some other nodes incident to them such that each node has degree $O(n^\delta)$ and the set of good nodes maintain their properties up to constant factors. 
    
    The second step computes an independent set on the sparsified graph as follows. It randomly samples nodes incident to good nodes with probability proportional to their degrees. Then, it adds to the independent set each sampled node that has no sampled neighbors. 
    To derandomize this process, we find a good random seed using a pessimistic estimators on the number of removed edges that is computed as follows. If a good node $v$ has exactly one sampled neighbor, then we optimistically add to the estimator its degree as if $v$ gets removed. Moreover, if a sampled node $u$ incident to $v$ has a sampled neighbor and thus is not independent, then $u$ subtracts the degree of $v$ on its behalf. Note that the degree of $v$ may be subtracted multiple times resulting in a pessimistic overall estimation. In the analysis, we show that, by restricting the set of (potentially sampled) neighbors for each good node $v$, the expected solution produced by this estimator removes a constant fraction of the edges.
    
    \subsection{Medium-Arboricity Case}
    
    This section is devoted to the $O(\log \arb \log \log_{\arb} n)$-bound  of \Cref{thm:mainlinear}. Whenever $\arb = o(2^{\log n / \log \log n})$, this algorithm achieves sublogarithmic running time. On our graph of remaining degree $\poly(\arb)$ by \Cref{thm:main_deg_red}, we run the following algorithm. For $i = 0, 1, \ldots, \log \log_\Delta n^\alpha + O(1)$:
    \begin{enumerate}
        \item Run the derandomized version of Luby's algorithm  for $2^i \cdot \Theta(\log \Delta)$ \local\ rounds using the derandomization of \Cref{sec:higharbcase} for $i = 0$ and that of~\cite[Section 5]{MPC-MIS-det-general} otherwise.
        \item Each vertex performs one graph exponentiation step until $\Delta^{2^i} = \localspace$.
    \end{enumerate}
    The proofs of \Cref{sec:mpc_part} and \cite[Section 5]{MPC-MIS-det-general} subsume the proof of the above algorithm. The key observation is that at any iteration $i$ the first step takes $O(\log \Delta) = O(\log \arb)$ \mpc\ rounds and ensures a decrease of the number of active vertices and edges by a $\Omega(\Delta^{2^i})$-factor and thus $\log \log_\Delta n^\alpha + O(1)$ iterations suffice.
    
    \subsection{Low-Arboricity Case}\label{sec:low-arb}
    The following lemma proves the last two running time bounds of \Cref{thm:mainlinear}, completing its proof. Recall that by \Cref{thm:main_deg_red} we can assume that our graph has remaining degree $\poly(\arb)$.
    \begin{lemma}\label{lm:low-arb}
        There is a deterministic low-memory \mpc\ algorithm that computes an \mis\ and an \mm\ of $G$ in time $O(\arb^{2+\eps} + \log \log n)$ and $O(\arb^{1+\eps} + \log \log n)$, respectively, for any constant $\eps > 0$, using linear total space. The same algorithm runs in time $O(\log \log n)$ when $\arb$ is $O(\log^{1/2-\delta} \log n)$ and $O(\log^{1-\delta} \log n)$ resp., for any constant $\delta > 0$.
    \end{lemma}
    \begin{proof}
        We first compute a degree $d = (\arb^{1+\epsilon})$ $H$-partition of the remaining graph $G$ in $O(\log \log n)$ rounds. Let $k$ be the largest integer such that $\Delta^{2^k} \le n^\alpha$, where $\Delta = \poly(\arb)$. For $i = 0,\ldots,k+O(1)$:
        \begin{enumerate}
            \item Nodes simulate $s \cdot 2^i$ rounds of the $H$-partition algorithm for suff. large $s = O(1)$.
            \item If $i < k$, then each node performs one step of graph exponentiation.
        \end{enumerate}
        The correctness and runtime and memory bounds of the above \mpc\ $H$-partition algorithm follow from~\cite[Lemma 4.2]{BFU18}, which is based on the observation that a $\Omega(\Delta)$ progress is achieved in $t = O(1)$ \local\ rounds by the properties of the $H$-partition. The \mis\ result follows from computing a $d^2 = O(\arb^{2+\varepsilon})$ coloring in $O(\log^* n)$ rounds by applying the algorithm \textsc{Arb-Linial's Coloring} of Barenboim and Elkin~\cite{BE10}.
        For the \mm\ case, the $d$ forest-decomposition produced by the $H$-partition is given in input to the maximal matching algorithm by Panconesi and Rizzi~\cite{PR01}, which runs in $O(d) = O(\arb^{1+\varepsilon})$ rounds. 
    \end{proof}

\section{Deterministic Arb-Coloring via All-to-All Communication}\label{chp:col}
    In this section, we study the complexity of deterministic arboricity-dependent coloring in all-to-all communication models, specifically, \congested\ and linear-memory \mpc, where each machine has $O(n)$ memory. In these models, Czumaj, Davis, and Parter~\cite{CDP20cc1} give a deterministic algorithm that solves $(\Delta+1)$ coloring, and its list-coloring variant, in $O(1)$ time. Henceforth, we refer to it as CDP's algorithm. 
    Towards our constant-round arboricity-dependent coloring algorithm, in the next corollary, we show that CDP's algorithm provides a $O(\arb)$ coloring of $\arb$-arboricity graphs in $O(\log \arb)$ rounds and a coloring with $O(\arb^{1+\eps})$ colors in $O(1)$ time.
    \begin{corollary}\label{cor:arbcol_basic}
        Given a graph $G$ with arboricity $\arb$ and a parameter $d$ satisfying $2\arb < d \le \arb^{1+\eps}$, for constant $\eps > 0$, a legal $O(d)$ vertex coloring of $G$ can be computed deterministically in $O(\log_{d/\arb} \arb)$ rounds of \congested\ and linear-memory \mpc\ with optimal global memory.
    \end{corollary}
    \begin{proof}
        The algorithm consists of three stages. The first stage computes the first $\last = O(\log_{d/\arb} \arb)$ layers of the $H$-partition of $G$ with degree $d$, and removes nodes in these $\last$ layers from the graph. The number of remaining nodes whose index layer is higher than $\last$ is at most $(\arb/d)^\last \cdot n$ by the definition of the $H$-partition. Therefore, at most $\lambda \cdot \left(\frac{\arb}{d}\right)^\last n = O(n)$ edges remain and the remaining graph can be gathered in a single machine and colored with $d+1$ colors locally.
        The second stage of the algorithm colors vertices in layers $\last,\ldots,2$ in a sequential fashion using $d$ additional colors. Consider an arbitrary layer $\ell \in [2,\last]$. Each node has a palette of $2d+1$ colors and at most $d$ neighbors already colored. Let each node discard colors blocked by its neighbors and observe that layer $\ell$ induces an instance of $(\Delta+1)$-list coloring solvable in $O(1)$ rounds by CDP's algorithm. The global memory required to store the color palettes is $n_\ell \cdot O(d)$, where $n_\ell$ is the number of nodes in layer $\ell$. Since $n_\ell \le \frac{\arb n}{d}$ for $\ell \ge 2$, the algorithm requires $\frac{\arb n}{d} \cdot O(d) = O(m)$ total space.
        In the third stage of the algorithm, we color nodes in the first layer using the constant-round $(\Delta+1)$-coloring algorithm of CDP and a fresh palette of $d+1$ colors, so that the global memory remains linear. The final coloring is a legal vertex coloring with $3d + 2 = O(d)$ colors, as desired.
    \end{proof}

    Next, we prove the main result of this section and its extension to \congested.  
    
    \mainarbcolalltoall*
         
    The first step of the algorithm is a preprocessing step. We compute the first $\last = O(1)$ layers of the $H$-partition of $G$ with degree $\Delta \eqdef \arb^{1 + \delta}$, for $0 < \delta \le \frac{2\eps}{1-2\eps}$. The number of edges induced by unlayered nodes is at most $O(n)$ by arboricity properties. Thus, we can gather and color unlayered vertices with $2\arb$ colors onto a single machine. In the rest of the algorithm, we discuss how to color the subgraph $H_j$ induced by the $j$-th layer, for $j \in [\last]$, with a fresh palette of $O(\arb)$ colors. 

    To color each remaining layer we perform a graph partitioning routine. Vertices are randomly partitioned in many subgraphs $B_1, \ldots, B_\ell$ of maximum \textit{out}-degree $O(\frac{\arb}{\ell})$ and size $O(\frac{n}{\ell})$. We show that such a partition exists by using the concept of pessimistic estimators. Then, a good partition is found by defining the global quality of a partition as a function of quantities computable by individual machines. However, doing so requires two steps. First, we select among all possible partitions those whose parts have \textit{all} linear size. Second, for each partition, we can now locally check how many \textit{bad nodes} are in each of its parts and sum these numbers up to find a good partition. Once a good partition is found, each of its $\ell$ subgraphs can be $O(\frac{\arb}{\ell})$-colored while the subgraph induced by bad nodes can be colored with $2\arb$ colors.
    
    \paragraph{The Graph Partitioning Algorithm:} Let $\ell = \arb^{0.6}$ be a parameter. The partition $V = B_1 \cup \cdots \cup B_\ell$ is obtained placing each node $v$ in a part selected according to a $k$-wise independent random choice. Denote $deg^{out}_B(v)$ as the outdegree of vertex $v$ in the subgraph induced by the vertices $B \subseteq V$. We require that the output of the partitioning algorithm satisfies the following two properties. For the sake of the analysis, let us imagine that we are given an orientation of the edges of $G$ with outdegree $2\arb$. 
 
    \subparagraph{i) Size of Each Part:} Bin $i$ in $[\ell]$ is \textit{good} if $\card{B_i} = O(\frac{n}{\ell})$. Observe that $\E[\card{B_i}] = \frac{n}{\ell} = \Omega(n^{0.4})$. Consequently, the probability that $\card{B_i}$ is at most $2\cdot \E[\card{B_i}]$ is upper bounded by $n^{-3}$, for $k \ge 16$, by applying \Cref{lem:kwise_bound} with $X_1,\ldots,X_n$ as i.r.v.\ for the events that the $i$-th vertex is in $B_i$. By a union bound over $\ell$ bins, with probability at least $1 - n^{-2}$, all bins contain fewer than $\frac{2n}{\ell}$ nodes.
        
    \subparagraph{ii) Outdegree of Each Vertex:} A node $v$ in bin $i$ is \textit{good} if $\deg^{out}_{B_i}(v) = O(\frac{\arb}{\ell})$. It is required that all $v \in V$ but at most $n/\arb^{3}$ nodes are good in expectation. We apply \Cref{lem:kwise_bound} with $X_1,\ldots,X_{deg^{out}(v)}$ as i.r.v.\ for the events that each outneighbor of $v$ is placed in the same bin as $v$. These variables are $(k-1)$-wise independent and each has expectation $\ell^{-1}$. Accordingly, the expected outdegree $\E[X]$ of $v$ is at most $\frac{2\arb}{\ell} = \Omega(\arb^{0.4})$. Therefore, its outdegree is larger than $\frac{4\arb}{\ell}$ with probability less than $\arb^{-3}$, for $k \ge 17$, and we can conclude the required property.

    \paragraph{Finding a Good Partition:} Let $\mathcal{H}$ be a $k$-wise independent family of hash functions $h \colon [n] \mapsto [\ell]$. Each hash function $h$ induces a partition of the graph into $\ell$ parts $B_1(h),\ldots,B_\ell(h)$. Similarly to~\cite{CDP20cc1}, we analyze the likelihood that a randomly chosen hash function $h$ induces any bad bin and many bad nodes by defining the following pessimistic estimator $P(h)$:
    \begin{equation*}
        P(h) = \card{\{\text{bad nodes under $h$}\}} + \card{\{\text{bad bins under $h$}\}}\cdot n.
    \end{equation*}
    Combining the bounds from the previous paragraph gives: 
    \begin{equation*}
        \E[P(h)] = \E[\card{\{\text{bad nodes}\}}] + \E[\card{\{\text{bad bins}\}}] \cdot n \le \frac{n}{\arb^3} + \frac{1}{n} \le \frac{2n}{\arb^3}.
    \end{equation*}
    Therefore, by the probabilistic method, there is one hash function $h^*$ that gives no bad bins and at most $\frac{2n}{\arb^3}$ bad nodes. 
    However, $P(h)$ cannot be easily computed without the orientation used for the analysis. 
    Due to that, we relax the above conditions and identify first a subset $\mathcal{H}' \subseteq \mathcal{H}$ such that $h^* \in \mathcal{H}'$ and each $h \in \mathcal{H}'$ gives a graph partitioning in which each subgraph has size $O(n)$. Second, we find a good partitioning by analyzing each hash function, i.e., induced subgraph, concurrently.
    
    \subparagraph{i) Select Candidate Functions:} Observe that $h^*$ is in $\mathcal{H}'$. Indeed, the bound on the size of each part together with the bound on the outdegree of each vertex gives
        \begin{equation*}
            \card{E(G[B_i(h)])} = \sum_{v \in B_i(h)} \deg_{B_i(h)}(v) \le \frac{2n}{\ell}\cdot\frac{4\arb}{\ell} + \frac{2n}{\arb^{3}}\cdot\Delta = O(n).
        \end{equation*}
    To deterministically find $\mathcal{H}'$, we assign each hash function $h \in \mathcal{H}$ to a machine $x_h$. Let each node compute its degree under $h$ and send it to $x_h$. The machine $x_h$ is then able to determine whether each subgraph induced by $h$ has size $O(n)$ by summing up the degrees of vertices in each particular bin. This process requires a total space of $O(\arb n \card{\mathcal{H}})$.
        
    \subparagraph{ii) Find a Good Function:} For every $h \in \mathcal{H}'$ and $i \in [\ell]$, we compute a $\frac{10\arb}{\ell}$-coloring of the subgraph $G[B_i(h)]$ and sum up the number of nodes that remain uncolored after this coloring step across all bins for each $h$. We prove that there is a function $h^* \in \mathcal{H}'$ that leaves only $O(n/\arb)$ nodes uncolored, which can be colored locally with a new palette of size $2\arb$. 
    Since $\card{G[B_i(h)]} = O(n)$, this computation can be run offline for each subgraph on a single machine with $O(\arb n \card{\mathcal{H}})$ total space. 
        
    Consider a subgraph $G[B_i(h)]$ for $h \in \mathcal{H}'$. This subgraph may have both good and bad nodes. 
    We run the $H$-partition algorithm with degree $d = \frac{10\arb}{\ell} - 1$ until when there is no vertex of degree at most $d$ left. 
    Once this process finishes, we obtain a partial $H$-partition $\mathcal{X}$ of $B_i(h)$ that can be $(d+1)$-colored. Then, nodes in $B_i'(h) \eqdef B_i(h) \setminus \mathcal{X}$ induce an uncolored subgraph of minimum degree at least $d+1$. This implies that $\card{E(G[B'_i(h)])} \ge \frac{5\arb}{\ell} \cdot \card{B'_i(h)}$. On the other hand, good nodes in $B_i'(h)$ contribute with at most $\frac{4\arb}{\ell}\cdot \card{B'_i(h)}$ edges and bad nodes with at most $\frac{2n}{\arb^3}\cdot \Delta$ edges. Thus,
    \begin{equation*}
        \frac{5\arb}{\ell} \cdot \card{B'_i(h)} \le \frac{4\arb}{\ell}\cdot \card{B'_i(h)} + \frac{2n}{\arb^3}\cdot \Delta \iff  \card{B'_i(h)} \le \frac{2\ell \Delta n}{\arb^4} \le \frac{n}{\arb^2},
    \end{equation*}
    which proves the existence of such good $h^* \in \mathcal{H}'$. We then find a partition whose bins can be $\frac{10\arb}{\ell}$-colored (with different palettes) leaving at most $\ell \cdot \frac{n}{\arb^2}$ nodes uncolored across all parts. Hence, uncolored vertices induce $\arb \cdot \frac{\ell n}{\arb^2} = O(n)$ edges in the original graph.

    \paragraph{Reducing the Total Space:} The space requirement of $O(\arb n \card{\mathcal{H}})$ can be made on the order of $n \cdot \poly(\arb)$ by reducing the size of $\mathcal{H}$ to be polynomial in the arboricity. To this end, we first note that w.l.o.g.\ one can assume that $\arb \le n^{\frac 1 2 - \eps}$, for any constant $\eps$, $0 < \eps < \frac 1 2$, as otherwise $\poly(n) = \poly(\arb)$. Second, we observe that $k$-wise independence is required only among the random variables associated with vertices that share a neighbor. An assignment of nodes to IDs such that no two nodes with a shared neighbor are assigned the same ID corresponds to a legal coloring of the subgraph $H_j^2$. In the subgraph $H_j^2$, each vertex has degree at most $\Delta^2 = \arb^{2 + 2\delta} = O(n)$, by our choice of $\delta$ and the assumption that $\arb \le n^{\frac 1 2 - \eps}$. Applying CDP's algorithm in parallel on each $H_j^2$ provides a $\poly(\arb)$ coloring in $O(1)$ time, as desired.\\

    Finally, we extend the above algorithm to \congested.
    \begin{corollary}\label{cc_extension_cor}
        When $\arb = O(n^{\delta}$), for constant $\delta > 0$, such that the $n \cdot \poly(\arb)$-global memory usage of \Cref{thm:main_arbcol_alltoall} can be made on the order of $n^2$, then the same algorithm works in \congested.
    \end{corollary}
    \begin{proof}
        This result follows from the equivalence between linear-memory \mpc\ and \congested. Specifically, each \mpc\ step can be simulated in a constant number of \congested\ rounds using Lenzen's routing~\cite{L13}. Since each of the $n$ nodes in \congested\ can communicate up to $O(n)$ words and our algorithm uses $O(n^2)$ total memory by the assumption on $\arb$, then \Cref{thm:main_arbcol_alltoall} can be run in \congested.
    \end{proof}

% \newpage
\bibliographystyle{alpha}
\bibliography{references}

\newcommand{\etalchar}[1]{$^{#1}$}
\begin{thebibliography}{LPSPP05}

\bibitem[ABI86]{ABI86}
Noga Alon, László Babai, and Alon Itai.
\newblock A fast and simple randomized parallel algorithm for the maximal
  independent set problem.
\newblock {\em Journal of Algorithms}, 7(4):567--583, 1986.

\bibitem[AS16]{alon2016probabilistic}
Noga Alon and Joel~H Spencer.
\newblock {\em The probabilistic method}.
\newblock John Wiley \& Sons, 2016.

\bibitem[BBD{\etalchar{+}}19]{MPC-MIS-log2logn}
Soheil Behnezhad, Sebastian Brandt, Mahsa Derakhshan, Manuela Fischer,
  MohammadTaghi Hajiaghayi, Richard~M. Karp, and Jara Uitto.
\newblock Massively parallel computation of matching and mis in sparse graphs.
\newblock In {\em Proceedings of the 2019 ACM Symposium on Principles of
  Distributed Computing}, PODC '19, page 481–490, New York, NY, USA, 2019.
  Association for Computing Machinery.

\bibitem[BDE{\etalchar{+}}19]{8948671}
Soheil Behnezhad, Laxman Dhulipala, Hossein Esfandiari, Jakub Lacki, and Vahab
  Mirrokni.
\newblock Near-optimal massively parallel graph connectivity.
\newblock In {\em 2019 IEEE 60th Annual Symposium on Foundations of Computer
  Science (FOCS)}, pages 1615--1636, 2019.

\bibitem[BDH18]{BDH18}
Soheil Behnezhad, Mahsa Derakhshan, and MohammadTaghi Hajiaghayi.
\newblock Brief announcement: Semi-mapreduce meets congested clique.
\newblock {\em CoRR}, abs/1802.10297, 2018.

\bibitem[BE10]{BE10}
Leonid Barenboim and Michael Elkin.
\newblock Sublogarithmic distributed mis algorithm for sparse graphs using
  nash-williams decomposition.
\newblock {\em Distributed Computing}, 22(5):363--379, 2010.

\bibitem[BE11]{BE11}
Leonid Barenboim and Michael Elkin.
\newblock Deterministic distributed vertex coloring in polylogarithmic time.
\newblock {\em J. ACM}, 58(5), 2011.

\bibitem[BE13]{BE13}
Leonid Barenboim and Michael Elkin.
\newblock Distributed graph coloring: Fundamentals and recent developments.
\newblock {\em Synthesis Lectures on Distributed Computing Theory},
  4(1):1--171, 2013.

\bibitem[BEPS16]{BEPS16}
Leonid Barenboim, Michael Elkin, Seth Pettie, and Johannes Schneider.
\newblock The locality of distributed symmetry breaking.
\newblock {\em J. ACM}, 63(3), 6 2016.

\bibitem[BFU18]{BFU18}
Sebastian Brandt, Manuela Fischer, and Jara Uitto.
\newblock Matching and {MIS} for uniformly sparse graphs in the low-memory
  {MPC} model.
\newblock {\em CoRR}, abs/1807.05374, 2018.

\bibitem[BFU19]{MPC-MIS-log3logn}
Sebastian Brandt, Manuela Fischer, and Jara Uitto.
\newblock Breaking the linear-memory barrier in $\mathsf{MPC}$: Fast
  $\mathsf{MIS}$ on trees with strongly sublinear memory.
\newblock In Keren Censor-Hillel and Michele Flammini, editors, {\em Structural
  Information and Communication Complexity}, pages 124--138, Cham, 2019.
  Springer International Publishing.

\bibitem[BK18]{BK18}
Leonid Barenboim and Victor Khazanov.
\newblock Distributed symmetry-breaking algorithms for congested cliques.
\newblock In Fedor~V. Fomin and Vladimir~V. Podolskii, editors, {\em Computer
  Science -- Theory and Applications}, pages 41--52, Cham, 2018. Springer
  International Publishing.

\bibitem[BKM20]{BKM20}
Philipp Bamberger, Fabian Kuhn, and Yannic Maus.
\newblock Efficient deterministic distributed coloring with small bandwidth.
\newblock In {\em Proceedings of the 39th Symposium on Principles of
  Distributed Computing}, PODC '20, page 243–252, New York, NY, USA, 2020.
  Association for Computing Machinery.

\bibitem[BR94]{BR94}
M.~Bellare and J.~Rompel.
\newblock Randomness-efficient oblivious sampling.
\newblock In {\em Proceedings 35th Annual Symposium on Foundations of Computer
  Science}, pages 276--287, 1994.

\bibitem[CC22]{coyconnectivity2021}
Sam Coy and Artur Czumaj.
\newblock Deterministic massively parallel connectivity.
\newblock In {\em Proceedings of the 54th Annual ACM SIGACT Symposium on Theory
  of Computing}, STOC 2022, page 162–175, New York, NY, USA, 2022.
  Association for Computing Machinery.

\bibitem[CDP20]{CDP20cc1}
Artur Czumaj, Peter Davies, and Merav Parter.
\newblock Simple, deterministic, constant-round coloring in the congested
  clique.
\newblock In {\em Proceedings of the 39th Symposium on Principles of
  Distributed Computing}, PODC '20, page 309–318, New York, NY, USA, 2020.
  Association for Computing Machinery.

\bibitem[CDP21a]{CDP21lb}
Artur Czumaj, Peter Davies, and Merav Parter.
\newblock Component stability in low-space massively parallel computation.
\newblock In {\em Proceedings of the 2021 ACM Symposium on Principles of
  Distributed Computing}, PODC'21, page 481–491, New York, NY, USA, 2021.
  Association for Computing Machinery.

\bibitem[CDP21b]{MPC-MIS-det-general}
Artur Czumaj, Peter Davies, and Merav Parter.
\newblock Graph sparsification for derandomizing massively parallel computation
  with low space.
\newblock {\em ACM Trans. Algorithms}, 17(2), 5 2021.

\bibitem[CDP21c]{CDP20ccb}
Artur Czumaj, Peter Davies, and Merav Parter.
\newblock Improved deterministic (delta+1) coloring in low-space mpc.
\newblock In {\em Proceedings of the 2021 ACM Symposium on Principles of
  Distributed Computing}, PODC'21, page 469–479, New York, NY, USA, 2021.
  Association for Computing Machinery.

\bibitem[CFG{\etalchar{+}}19]{CFG+19}
Yi-Jun Chang, Manuela Fischer, Mohsen Ghaffari, Jara Uitto, and Yufan Zheng.
\newblock The complexity of (delta+1) coloring in congested clique, massively
  parallel computation, and centralized local computation.
\newblock In {\em Proceedings of the 2019 ACM Symposium on Principles of
  Distributed Computing}, PODC '19, page 471–480, New York, NY, USA, 2019.
  Association for Computing Machinery.

\bibitem[CG89]{CG89}
Benny Chor and Oded Goldreich.
\newblock On the power of two-point based sampling.
\newblock {\em Journal of Complexity}, 5(1):96--106, 1989.

\bibitem[CHPS20]{CPS17}
Keren Censor-Hillel, Merav Parter, and Gregory Schwartzman.
\newblock Derandomizing local distributed algorithms under bandwidth
  restrictions.
\newblock {\em Distributed Computing}, 33(3):349--366, Jun 2020.

\bibitem[CLM{\etalchar{+}}18]{CLM+18}
Artur Czumaj, Jakub Lacki, Aleksander Madry, Slobodan Mitrovi\'{c}, Krzysztof
  Onak, and Piotr Sankowski.
\newblock Round compression for parallel matching algorithms.
\newblock In {\em Proceedings of the 50th Annual ACM SIGACT Symposium on Theory
  of Computing}, STOC 2018, page 471–484, New York, NY, USA, 2018.
  Association for Computing Machinery.

\bibitem[CW79]{CW79}
J.Lawrence Carter and Mark~N. Wegman.
\newblock Universal classes of hash functions.
\newblock {\em Journal of Computer and System Sciences}, 18(2):143--154, 1979.

\bibitem[DG08]{MapReduce}
Jeffrey Dean and Sanjay Ghemawat.
\newblock Mapreduce: Simplified data processing on large clusters.
\newblock {\em Commun. ACM}, 51(1):107–113, 1 2008.

\bibitem[EFF85]{EFF85}
Paul Erd{\H{o}}s, Peter Frankl, and Zolt{\'a}n F{\"u}redi.
\newblock Families of finite sets in which no set is covered by the union of r
  others.
\newblock {\em Israel J. Math}, 51(1-2):79--89, 1985.

\bibitem[EGL{\etalchar{+}}98]{EGL+98}
Guy Even, Oded Goldreich, Michael Luby, Noam Nisan, and Boban Veličković.
\newblock Efficient approximation of product distributions.
\newblock {\em Random Structures \& Algorithms}, 13(1):1--16, 1998.

\bibitem[FGG22]{FGG22P}
Manuela Fischer, Jeff Giliberti, and Christoph Grunau.
\newblock {Improved Deterministic Connectivity in Massively Parallel
  Computation}.
\newblock In Christian Scheideler, editor, {\em 36th International Symposium on
  Distributed Computing (DISC 2022)}, volume 246 of {\em Leibniz International
  Proceedings in Informatics (LIPIcs)}, pages 22:1--22:17, Dagstuhl, Germany,
  2022. Schloss Dagstuhl -- Leibniz-Zentrum f{\"u}r Informatik.

\bibitem[FMS{\etalchar{+}}10]{FMS+10}
Jon Feldman, S.~Muthukrishnan, Anastasios Sidiropoulos, Cliff Stein, and Zoya
  Svitkina.
\newblock On distributing symmetric streaming computations.
\newblock {\em ACM Trans. Algorithms}, 6(4), 9 2010.

\bibitem[GGJ20]{MPC-MIS-loglogn}
Mohsen Ghaffari, Christoph Grunau, and Ce~Jin.
\newblock {Improved MPC Algorithms for MIS, Matching, and Coloring on Trees and
  Beyond}.
\newblock In Hagit Attiya, editor, {\em 34th International Symposium on
  Distributed Computing (DISC 2020)}, volume 179 of {\em Leibniz International
  Proceedings in Informatics (LIPIcs)}, pages 34:1--34:18, Dagstuhl, Germany,
  2020. Schloss Dagstuhl--Leibniz-Zentrum f{\"u}r Informatik.

\bibitem[GGK{\etalchar{+}}18]{GGKMR18}
Mohsen Ghaffari, Themis Gouleakis, Christian Konrad, Slobodan Mitrovi\'{c}, and
  Ronitt Rubinfeld.
\newblock Improved massively parallel computation algorithms for mis, matching,
  and vertex cover.
\newblock In {\em Proceedings of the 2018 ACM Symposium on Principles of
  Distributed Computing}, PODC '18, page 129–138, New York, NY, USA, 2018.
  Association for Computing Machinery.

\bibitem[Gha16]{Gha16}
Mohsen Ghaffari.
\newblock An improved distributed algorithm for maximal independent set.
\newblock In {\em Proceedings of the 2016 Annual ACM-SIAM Symposium on Discrete
  Algorithms (SODA)}, pages 270--277, 2016.

\bibitem[GKU19]{GKU19}
Mohsen Ghaffari, Fabian Kuhn, and Jara Uitto.
\newblock Conditional hardness results for massively parallel computation from
  distributed lower bounds.
\newblock In {\em 2019 IEEE 60th Annual Symposium on Foundations of Computer
  Science (FOCS)}, pages 1650--1663, 2019.

\bibitem[GL17]{GL17}
Mohsen Ghaffari and Christiana Lymouri.
\newblock {Simple and Near-Optimal Distributed Coloring for Sparse Graphs}.
\newblock In Andr{\'e}a~W. Richa, editor, {\em 31st International Symposium on
  Distributed Computing (DISC 2017)}, volume~91 of {\em Leibniz International
  Proceedings in Informatics (LIPIcs)}, pages 20:1--20:14, Dagstuhl, Germany,
  2017. Schloss Dagstuhl--Leibniz-Zentrum fuer Informatik.

\bibitem[Goo99]{G99}
Michael~T. Goodrich.
\newblock Communication-efficient parallel sorting.
\newblock {\em SIAM Journal on Computing}, 29(2):416--432, 1999.

\bibitem[GS19]{GS19}
Mohsen Ghaffari and Ali Sayyadi.
\newblock {Distributed Arboricity-Dependent Graph Coloring via All-to-All
  Communication}.
\newblock In Christel Baier, Ioannis Chatzigiannakis, Paola Flocchini, and
  Stefano Leonardi, editors, {\em 46th International Colloquium on Automata,
  Languages, and Programming (ICALP 2019)}, volume 132 of {\em Leibniz
  International Proceedings in Informatics (LIPIcs)}, pages 142:1--142:14,
  Dagstuhl, Germany, 2019. Schloss Dagstuhl--Leibniz-Zentrum fuer Informatik.

\bibitem[GSZ11]{GSZ11}
Michael~T. Goodrich, Nodari Sitchinava, and Qin Zhang.
\newblock Sorting, searching, and simulation in the mapreduce framework.
\newblock In Takao Asano, Shin-ichi Nakano, Yoshio Okamoto, and Osamu Watanabe,
  editors, {\em Algorithms and Computation}, pages 374--383, Berlin,
  Heidelberg, 2011. Springer Berlin Heidelberg.

\bibitem[GU19]{doi:10.1137/1.9781611975482.99}
Mohsen Ghaffari and Jara Uitto.
\newblock Sparsifying distributed algorithms with ramifications in massively
  parallel computation and centralized local computation.
\newblock In {\em Proceedings of the 2019 Annual ACM-SIAM Symposium on Discrete
  Algorithms (SODA)}, pages 1636--1653, 2019.

\bibitem[HP15]{HP15}
James~W. Hegeman and Sriram~V. Pemmaraju.
\newblock Lessons from the congested clique applied to mapreduce.
\newblock {\em Theoretical Computer Science}, 608:268--281, 2015.
\newblock Structural Information and Communication Complexity.

\bibitem[II86]{II86}
Amos Israeli and A.~Itai.
\newblock A fast and simple randomized parallel algorithm for maximal matching.
\newblock {\em Information Processing Letters}, 22(2):77--80, 1986.

\bibitem[KPP20]{KPP20}
Kishore Kothapalli, Shreyas Pai, and Sriram~V. Pemmaraju.
\newblock {Sample-And-Gather: Fast Ruling Set Algorithms in the Low-Memory MPC
  Model}.
\newblock In Nitin Saxena and Sunil Simon, editors, {\em 40th IARCS Annual
  Conference on Foundations of Software Technology and Theoretical Computer
  Science (FSTTCS 2020)}, volume 182 of {\em Leibniz International Proceedings
  in Informatics (LIPIcs)}, pages 28:1--28:18, Dagstuhl, Germany, 2020. Schloss
  Dagstuhl--Leibniz-Zentrum f{\"u}r Informatik.

\bibitem[KSV10]{KSV10}
Howard Karloff, Siddharth Suri, and Sergei Vassilvitskii.
\newblock A model of computation for mapreduce.
\newblock In {\em Proceedings of the 2010 Annual ACM-SIAM Symposium on Discrete
  Algorithms (SODA)}, pages 938--948, 2010.

\bibitem[KSV13]{KSV13}
Amos Korman, Jean-S{\'e}bastien Sereni, and Laurent Viennot.
\newblock Toward more localized local algorithms: removing assumptions
  concerning global knowledge.
\newblock {\em Distributed Computing}, 26(5):289--308, 2013.

\bibitem[Kuh09]{10.1145/1583991.1584032}
Fabian Kuhn.
\newblock Weak graph colorings: Distributed algorithms and applications.
\newblock In {\em Proceedings of the Twenty-First Annual Symposium on
  Parallelism in Algorithms and Architectures}, SPAA '09, page 138–144, New
  York, NY, USA, 2009. Association for Computing Machinery.

\bibitem[Len13]{L13}
Christoph Lenzen.
\newblock Optimal deterministic routing and sorting on the congested clique.
\newblock In {\em Proceedings of the 2013 ACM Symposium on Principles of
  Distributed Computing}, PODC '13, page 42–50, New York, NY, USA, 2013.
  Association for Computing Machinery.

\bibitem[Lin92]{doi:10.1137/0221015}
Nathan Linial.
\newblock Locality in distributed graph algorithms.
\newblock {\em SIAM Journal on Computing}, 21(1):193--201, 1992.

\bibitem[LPSPP05]{Congested}
Zvi Lotker, Boaz Patt-Shamir, Elan Pavlov, and David Peleg.
\newblock Minimum-weight spanning tree construction in o(log log n)
  communication rounds.
\newblock {\em SIAM Journal on Computing}, 35(1):120--131, 2005.

\bibitem[LU21]{MPC-MIS-det-log2logn}
Rustam Latypov and Jara Uitto.
\newblock Deterministic 3-coloring of trees in the sublinear {MPC} model.
\newblock {\em CoRR}, abs/2105.13980, 2021.

\bibitem[Lub85]{L85}
M~Luby.
\newblock A simple parallel algorithm for the maximal independent set problem.
\newblock In {\em Proceedings of the Seventeenth Annual ACM Symposium on Theory
  of Computing}, STOC '85, page 1–10, New York, NY, USA, 1985. Association
  for Computing Machinery.

\bibitem[Lub93]{Lub93}
Michael Luby.
\newblock Removing randomness in parallel computation without a processor
  penalty.
\newblock {\em Journal of Computer and System Sciences}, 47(2):250--286, 1993.

\bibitem[LW06]{LW06}
Michael Luby and Avi Wigderson.
\newblock Pairwise independence and derandomization.
\newblock {\em Foundations and Trends in Theoretical Computer Science},
  1(4):237--301, 2006.

\bibitem[LW10]{LW10}
Christoph Lenzen and Roger Wattenhofer.
\newblock Brief announcement: Exponential speed-up of local algorithms using
  non-local communication.
\newblock In {\em Proceedings of the 29th ACM SIGACT-SIGOPS Symposium on
  Principles of Distributed Computing}, PODC '10, page 295–296, New York, NY,
  USA, 2010. Association for Computing Machinery.

\bibitem[MR95]{MR95}
Rajeev Motwani and Prabhakar Raghavan.
\newblock {\em Randomized algorithms}.
\newblock Cambridge university press, 1995.

\bibitem[MSN94]{MNN94}
Rajeev Motwani, Joseph {(Seffi) Naor}, and Moni Naor.
\newblock The probabilistic method yields deterministic parallel algorithms.
\newblock {\em Journal of Computer and System Sciences}, 49(3):478--516, 1994.
\newblock 30th IEEE Conference on Foundations of Computer Science.

\bibitem[NW61]{NW61}
C.~St.J.~A. Nash-Williams.
\newblock {Edge-Disjoint Spanning Trees of Finite Graphs}.
\newblock {\em Journal of the London Mathematical Society}, s1-36(1):445--450,
  01 1961.

\bibitem[NW64]{NW64}
C.~St.J.~A. Nash-Williams.
\newblock {Decomposition of Finite Graphs Into Forests}.
\newblock {\em Journal of the London Mathematical Society}, s1-39(1):12--12, 01
  1964.

\bibitem[Par18]{P18}
Merav Parter.
\newblock {(Delta+1) Coloring in the Congested Clique Model}.
\newblock In Ioannis Chatzigiannakis, Christos Kaklamanis, D{\'a}niel Marx, and
  Donald Sannella, editors, {\em 45th International Colloquium on Automata,
  Languages, and Programming (ICALP 2018)}, volume 107 of {\em Leibniz
  International Proceedings in Informatics (LIPIcs)}, pages 160:1--160:14,
  Dagstuhl, Germany, 2018. Schloss Dagstuhl--Leibniz-Zentrum fuer Informatik.

\bibitem[PP22]{PP22}
Shreyas Pai and Sriram~V. Pemmaraju.
\newblock Brief announcement: Deterministic massively parallel algorithms for
  ruling sets.
\newblock In {\em Proceedings of the 2022 ACM Symposium on Principles of
  Distributed Computing}, PODC'22, page 366–368, New York, NY, USA, 2022.
  Association for Computing Machinery.

\bibitem[PR01]{PR01}
Alessandro Panconesi and Romeo Rizzi.
\newblock Some simple distributed algorithms for sparse networks.
\newblock {\em Distributed computing}, 14(2):97--100, 2001.

\bibitem[PS18]{PS18}
Merav Parter and Hsin-Hao Su.
\newblock {Randomized (Delta+1)-Coloring in O(log* Delta) Congested Clique
  Rounds}.
\newblock In Ulrich Schmid and Josef Widder, editors, {\em 32nd International
  Symposium on Distributed Computing (DISC 2018)}, volume 121 of {\em Leibniz
  International Proceedings in Informatics (LIPIcs)}, pages 39:1--39:18,
  Dagstuhl, Germany, 2018. Schloss Dagstuhl--Leibniz-Zentrum fuer Informatik.

\bibitem[Rag88]{Rag88}
Prabhakar Raghavan.
\newblock Probabilistic construction of deterministic algorithms: Approximating
  packing integer programs.
\newblock {\em Journal of Computer and System Sciences}, 37(2):130--143, 1988.

\bibitem[WC79]{WC79}
Mark~N. Wegman and J.~Lawrence Carter.
\newblock New classes and applications of hash functions.
\newblock In {\em 20th Annual Symposium on Foundations of Computer Science
  (sfcs 1979)}, pages 175--182, 1979.

\bibitem[ZCF{\etalchar{+}}10]{Spark}
Matei Zaharia, Mosharaf Chowdhury, Michael~J. Franklin, Scott Shenker, and Ion
  Stoica.
\newblock Spark: Cluster computing with working sets.
\newblock In {\em Proceedings of the 2nd USENIX Conference on Hot Topics in
  Cloud Computing}, HotCloud'10, page~10, USA, 2010. USENIX Association.

\end{thebibliography}
\clearpage
\appendix
\section{Arboricity Properties of Graphs and H-partition}\label{sec:graph_arboricity}
     The notion of arboricity characterizes the sparsity of a graph and has the following properties. Nash-Williams~\cite{NW61, NW64} showed that the arboricity equals the density of the graph, i.e., the ratio of the number of edges to that of nodes in any subgraph.
    \begin{lemma}[Arboricity Properties \cite{BE13, BEPS16}] \label{lem:arb}
    Let $G$ be a graph of $m$ edges, $n$ nodes, and arboricity $\arb$. Then
        \begin{enumerate}
            \item $m < \arb n$.
            \label{arb1}
            \item The arboricity of any subgraph $G' \subseteq G$ is at most $\arb$.
            \item The number of nodes with degree at least $t\ge \arb+1$ is less than $\arb n/(t-\arb)$.
            \label{arb2}
            \item The number of edges whose endpoints both have degree at least $t\ge \arb+1$ is less than $\arb m/(t-\arb)$.
            \label{arb3}
        \end{enumerate}
    \end{lemma}
    An $H$-partition (see \Cref{def:hpartition}) can be quickly constructed as follows, making use of the above properties.
    Each layer is computed in a constant number of rounds by applying the classical greedy peeling algorithm as follows: Repeatedly, vertices of remaining degree at most $d$ are peeled off from the graph and form a new layer. Since the number of vertices of degree larger than $d$ is at most $\frac{2\arb}{d} \cdot n$, the \textit{size} $L$ of $\hpart$, i.e., the overall number of layers, is at most $\ceil{\log_{d/2\arb} n} = O(\log_{\frac{d}{\arb}} n)$. Given an $H$-partition, a $d$ forest decomposition, i.e., an acyclic orientation of the edges such that every vertex has outdegree at most $d$, can be obtained as follows. An edge $(u,v)$ is oriented toward the vertex with a higher index layer in the given $H$-partition or, if they are in the same layer, toward the vertex with a greater ID. By construction, the orientation is acyclic and each vertex has at most $d$ outneighbors. Now, let each node $v$ assign a distinct label from the set $\{1,\ldots,d\}$ to each of its outgoing edges. The set of edges labeled by label $i$ forms a forest. For a more detailed proof, we refer to~\cite{BE13}.

    \begin{remark}
        While our algorithms may seem to require knowing $\arb$, by employing the standard technique of~\cite{KSV13} to run the algorithm with doubly-exponentially increasing estimates for $\arb$ in parallel, combined with global communication, we can find an estimate for $\arb$ that produces the same asymptotic result incurring only a $\tilde O(1)$ factor overhead in the global space and total computation.
    \end{remark}

\section{Derandomization via All-to-All Communication}\label{sec:derandomization_framework}
    A $k$-wise independent family of hash functions is defined as follows:
    \begin{definition}[$k$-wise independence]
        Let $N, k, \ell \in \mathbb{N}$ with $k \le N$. A family of hash functions $\mathcal{H} = \{h : [N] \rightarrow \{0, 1\}^{\ell}\}$ is $k$-wise independent if for all $I \subseteq \{1,\ldots,n\}$ with $|I| \leq k$, the random variables $X_i := h(i)$\footnote{$h(\cdot)$ denotes the length-$\ell$ bit sequence by the corresponding integer in $\{0,\ldots,2^\ell-1\}$.} with $i \in I$ are independent and uniformly distributed in $\{0, 1\}^\ell$, when $h$ is chosen uniformly at random from $\mathcal{H}$. If $k = 2$ then $\mathcal{H}$  is called pairwise independent.
    \end{definition}
    Given a randomized process that works under $k$-wise independence, we apply the following lemma to construct a $k$-wise independent family of hash functions $\mathcal{H}$.
    \begin{lemma}[\cite{ABI86, CG89, EGL+98}]\label{lemma-hash}
        For every $N, k, \ell \in \mathbb{N}$, there is a family of $k$-wise independent hash functions $\mathcal{H} = \{h : [N] \rightarrow \{0, 1\}^{\ell}\}$ such that choosing a uniformly random function $h$ from $\mathcal{H}$ takes at most $k(\ell + \log N) + O(1)$ random bits, and evaluating a function from $\mathcal{H}$ takes time $\poly(\ell, \log N)$ time.
    \end{lemma}
    For a more systematic treatment of limited independence and its applications, we refer the reader to~\cite{alon2016probabilistic, CW79, LW06, MR95, Rag88, WC79}. We now define an objective function that is a sum of functions calculable by individual machines and is at least (or at most) some value $Q$. Its expected value over the choice of a random hash function $h \in \mathcal{H}$ is then
    \begin{equation*}
        \E_{h \in \mathcal{H}}\left[q(h) \: \eqdef \sum_{\text{machines } x} q_{x}(h)\right] \ge Q.
    \end{equation*}
    By the probabilistic method, this implies the existence of a hash function $h^* \in \mathcal{H}$ for which $q(h^*)$ is at least (or at most) $Q$. If the family of hash functions fits into the memory of a single machine, then each machine $x$ can locally compute $q(h)$ (which is a function of the vertices that machine $x$ is assigned) for each $h$. Then, we sum the values $q(h)$ computed by every machine for each $h$ and find one good hash function $h^*$ in $O(1)$ \mpc\ rounds. When the family of hash functions cannot be collected onto a single machine but its size is $\poly(n)$, i.e., each function can be specified using $O(\log n)$ bits, we can still find one good hash function in $O(1)$ \mpc\ rounds as follows. 
    
    Let $\tau = O(\log n)$ denote the length of the random seed and $\chi = \log \memory = \alpha \log n$. We proceed in $\ceil{\frac{\tau}{\chi}}$ iterations. We search for a good function $h^*$ by having machines agree on a $\chi$-bit chunk at a time. Let $\mathcal{H}^{\langle i \rangle}$ be the subset of hash functions from $\mathcal{H}$ considered at iteration $i$. Clearly, $\mathcal{H}^{\langle 0 \rangle} = \mathcal{H}$ and $\mathcal{H}^{\langle i \rangle} \supseteq \mathcal{H}^{\langle i+1 \rangle}$. We will ensure that $\E_{h \in \mathcal{H}^{\langle i \rangle}}\left[q(h)\right] \ge \E_{h \in \mathcal{H}^{\langle i-1 \rangle}}\left[q(h)\right]$ until $\mathcal{H}^{\langle i \rangle}$ consists of a single element defining $h^*$. 
    In phase $i$, for $1 \le i \le \ceil{\frac{\tau}{\chi}}$, we~determine $\mathcal{H}^{\langle i \rangle}$ by fixing the bits at positions $1 + (i-1)\ceil{\frac{\tau}{\chi}}, \ldots, i\ceil{\frac{\tau}{\chi}}$. We assume that each machine knows the prefix $\mathbf{b}$ chosen in previous iterations. Next, each machine considers the extension of $\mathbf{b}$ by all possible seeds of length $\chi$. Let $\mathcal{H}_1, \ldots, \mathcal{H}_{2^\chi}$ be a partition of $\mathcal{H}^{\langle i-1 \rangle}$ such that $\mathcal{H}_j$ contains all hash functions $h \in \mathcal{H}^{\langle i-1 \rangle}$ whose bit sequence consists of $b$ as prefix followed by the bit representation $b_j$ of $j$. 
    
    For every $\mathcal{H}_j$, each machine computes locally $\E_{h \in \mathcal{H}_j}\left[q(h)\right]$ for all $1 \le j \le 2^{\chi}$. This computation can be performed \textit{locally} by taking each hash function $h \in \mathcal{H}_j$ and computing $q_{x}(h)$, which is the target function for the nodes (or edges) that machine $x$ is responsible for. Once we know $q_{x}(h)$ for every $h \in \mathcal{H}_j$, then $\E_{h \in \mathcal{H}_j}\left[q_{x}(h)\right]$ is simply the average of all these numbers. 
    A more efficient way to compute $\E_{h \in \mathcal{H}_j}\left[q_{x}(h)\right]$ is by letting each machine compute the expected value of $q_{x}(h)$ for each of its vertices (or edges) conditioned on $\mathbf{b} \circ b_j$ (note that each node $u$ can compute this probability if it depends only on its neighbors and, by knowing their IDs, $u$ can simulate their choices).
    Using \Cref{lem:primitives}, we compute $\sum_{\text{machines } x} \E_{h \in \mathcal{H}_j}\left[q_{x}(h)\right]$ by summing up the individual expectations. Note that by linearity of expectations, it holds that $\sum_{\text{machines } x} \E_{h \in \mathcal{H}_j}\left[q_{x}(h)\right] = \E_{h \in \mathcal{H}_j}\left[q(h)\right]$.
    Then, we are ensured that there is an index $k$ such that $\E_{h \in \mathcal{H}_k}\left[q(h)\right] \ge Q$, i.e., picking a u.a.r. function $h$ from $\mathcal{H}_k$ gives an objective that is at least (or at most) $Q$. Thus, we select $\mathcal{H}^{\langle i \rangle} = \mathcal{H}_k$ and extend the bit prefix $\mathbf{b} \gets \mathbf{b} \circ b_k$ with the bit representation $b_k$ of $k$. In the next iteration, we repeat the same process until when one good hash function $h^*$ is found. Since $\lceil \frac{\tau}{\chi} \rceil = O(1/\alpha) = O(1)$, this distributed implementation of the method of conditional expectation runs in $O(1)$ \mpc\ rounds.

    \paragraph{Concentration Inequalities} To show that a randomized process requires only $k$-wise independence, we use the following results. For $k = 2$, a bound that is usually satisfactory is obtained from Chebyshev’s inequality. For larger $k$, we apply the subsequent tail inequality, which depends on $k$.
    \begin{theorem}[Chebyshev]\label{thm:cheb}
        Let $X_1,\ldots,X_n$ be random variables taking values in $[0,1]$. Let $X = X_1 + \ldots + X_n$ denote their sum and $\mu \le \E[X]$. Then
        \begin{equation*}
            \Pr\left[|X - \E[X]| \ge \E[X]\right] \le \frac{\Var[X]}{\mu^2}.
        \end{equation*}
    \end{theorem}
    \begin{corollary}\label{cor:chebpw}
        If $X_1,\ldots,X_n$ are pairwise independent, then $\Var[X] = \sum_{i=1}^n \Var[X_i] \le \E[X]$ and
        \begin{equation*}
            \Pr\left[|X - \E[X]| \ge \E[X]\right] \le \frac{\sum_{i=1}^n \Var[X_i]}{\E[X]^2} \le \frac{1}{\E[X]} \le \frac{1}{\mu}.
        \end{equation*}
    \end{corollary}
    \begin{lemma}[Lemma 2.3 of~\cite{BR94}]\label{lem:kwise_bound}
        Let $k \ge 4$ be an even integer. Let $X_1,\ldots,X_n$ be random variables taking values in $[0,1]$. Let $X = X_1 + \ldots + X_n$ denote their sum and let $\mu \le \E[X]$ satisfying $\mu \ge k$. Then, for any $\epsilon > 0$, we have
        \begin{equation*}
            \Pr\left[|X - \E[X]| \ge \epsilon \cdot \E[X]\right] \le 8 \left( \frac{2k}{\epsilon^2 \mu}\right)^{k/2}.
        \end{equation*}
    \end{lemma}

\section{Missing Proofs of Arboricity-Based Degree Reduction}\label{sec:missing_proofs}
    \begin{proof}[Proof of \Cref{lem:degree_reduction_extract_single}]
        Let $\mathcal{J} = \{v \in V(G,\Delta_{i}) \;|\; \deg_{V(G,\Delta_{i})}(v) \ge \deg_G(v)/2\}$ denote high-degree nodes having a majority of their neighbors in $ V(G,\Delta_{i})$.  
        Then, let each node $v \in  V(G,\Delta_{i})' \:\eqdef\:  V(G,\Delta_{i}) \setminus \mathcal{J}$, which has $\deg_{V\setminus V(G,\Delta_{i})}(v) \ge \deg_G(v)/2 \ge \highdeg/2$, discard some edges arbitrarily such that $\deg_{V\setminus V(G,\Delta_{i})}(v) = 2\beta^8$. Let $\tilde{E}$ be the set of remaining edges crossing the cut $(V(G,\Delta_{i}),V\setminus V(G,\Delta_{i}))$. Let $\Shell = \{u \mid v\in V(G,\Delta_{i})' \mbox{ and } \{u,v\} \in \tilde{E}\}$ be the low-degree neighborhood of $V(G,\Delta_{i})'$. We define good $\Shell$-nodes and good $V(G,\Delta_{i})'$-nodes as follows.
        \begin{align*}
            &\Shellgood = \left\{u \in \Shell \mid \deg_{\tilde{E}}(u) < \beta, \deg_{\Shell}(u) < \beta^2\right\},\text{ and}\\
            &V(G,\Delta_{i})_{\textsf{Good}} = \left\{v \in V(G,\Delta_{i})' \mid \deg_{(\Shellgood, \tilde{E})}(v) \ge \beta^8 \right\}.
        \end{align*}
        Then, each node $v \in V(G,\Delta_{i})_{\textsf{Good}}$ arbitrarily selects a a subset $\goodset{v}$ of size $\beta^4$ of $\Shellgood$-nodes it is incident to. Thus, the $\degreehigh$-high-subgraph of $G$ is given by the vertices $\left( V(G,\Delta_{i})_{\textsf{Good}} \sqcup \Shellgood\right) =: (V^{high}, V^{low})$ and the two set of edges $$E_1 = \{(u,v) \in E \mid v \in V(G,\Delta_{i})_{\textsf{Good}}, u \in \goodset{v}\}$$ and $$E_2 = \{(u,u') \in E \mid u, u' \in \bigcup_{v \in V(G,\Delta_{i})_{\textsf{Good}}} T(v)\}.$$ Finally, the number of \textit{bad} nodes that are in $\mathcal{J}$ or $V(G,\Delta_{i})' \setminus V(G,\Delta_{i})_{\textsf{Good}}$ is upper bounded by $\frac{5 \arb |V(G,\Delta_{i})|}{\beta-\arb} \le |V(G,\Delta_{i})|/\Delta^{\Omega(1)}_i$ by the analysis of \cite{BEPS16} and our choice of $\beta = \Delta^{\Omega(1)}$ with $\Delta_i \geq \poly(\arb)$. 
    \end{proof}
    \begin{proof}[Proof of \Cref{lem:degree_reduction_multi_multi}]
        We apply the algorithm of \Cref{lem:degree_reduction_multi_single} for $k$ times in a sequential fashion. 
        Let $G^{(0)} = G$ and $G^{(\ell)}$ denote the graph obtained after the $\ell$-th execution of \Cref{lem:degree_reduction_multi_single}, for $0 \le \ell \le k$. 
        We construct $G^{(\ell)}$ by running the algorithm on $G^{(\ell-1)}$ and removing nodes in the computed independent set and their neighbors (or nodes incident to the matching) from $G^{(\ell-1)}$. Then, $G' = G^{(k)}$.
        The overall running time is $O(k)$ since each run takes $O(1)$ rounds. We let each run use $O(\log C)$ different random bits so that each run succeeds independently with strictly positive probability by \Cref{lem:degree_reduction_multi_single}. It then follows that the probability that all the runs succeed is strictly positive and that the total number of shared random bits required is $O(k \cdot \log C)$. It remains to prove the upper bound on the size of $|V_{\ge i}(G')|$ in the resulting graph $G'$. 
        The proof is by induction on $\ell$. The base case with $\ell = 0$ clearly holds. Assume that the statement holds for $\ell = k-1$ and fix an arbitrary $i \in [i_{min}, i_{max}]$. By \Cref{lem:degree_reduction_multi_single} for some constant $c > 0$, we obtain
        \begin{align*}
            |V_{\ge i}(G^{(\ell + 1)})| &\le \frac{|V_{i}(G^{(\ell)})|}{\Delta_i^c} + \sum_{j > i} \Delta_j^{O(1)} |V_j(G^{(\ell)})| 
            \\ & \le \frac{|V_{\ge i}(G^{(\ell)})|}{\Delta_i^{c\cdot k}} \\ & \quad + \sum_{j > i} \Delta_j^{O(1)} \frac{|V_{\ge j}(G^{(\ell)})|}{\Delta_j^{c\cdot(k-1)}}\\
            &\le \frac{|V_{\ge i}(G^{(\ell)})|}{\Delta_i^{c\cdot k}} + \\ & \quad |V_{\ge i}(G^{(\ell)})| \sum_{j > i} \Delta_i^{-(c\cdot(k-1) - O(1))(2^j - 2^i)} \\
            &\le \frac{|V_{\ge i}(G^{(\ell)})|}{\Delta_i^{\Omega(k)}}\left(1 + \sum_{j \ge i} \Delta_i^{(2^i - 2^j + 1)} \right) = \frac{|V_{\ge i}(G^{(\ell)})|}{\Delta_i^{\Omega(k)}}.
        \end{align*}
    \end{proof}
    \begin{proof}[Proof of \Cref{lem:degree_reduction_mpc_single_round}]
        We present an $O(1)$-time algorithm that achieves a $\Delta^{\Omega(1)}_i$ reduction of $|V_{\ge i}(G)|$. The lemma then follows from repeating it $O(c)$ times. 
        The first step of the algorithm applies \cref{lem:degree_reduction_extract_single} with input $G$ and $\Delta_i$ to obtain a $\degreehigh$-high-graph $H$ with $\degreehigh = \Delta_i^{\Omega(1)}$. Assuming that each node can store its one-hop neighbors on a single machine since $\Delta = O(n^\delta)$, \cref{lem:degree_reduction_extract_single} can be simulated on a low-memory \mpc\ in $O(1)$ rounds. Observe that $V_{\ge i}(G) = \{v \in V(G) \mid \deg_{G}(v) > \sqrt{\Delta_i}\} =: V(G,\Delta_i)$ of \cref{lem:degree_reduction_extract_single}. By Property~\ref{local_single_A}~and~\ref{local_single_C} of~\cref{lem:degree_reduction_extract_single}, $|V(G,\Delta_{i}) \setminus V^{high}(H)| \leq \frac{|V(G,\Delta_{i})|}{\Delta^{\Omega(1)}_i}$ and $\card{V^{high}(H)} \le \card{V_{\ge i}(G)}$.
        Thus, it remains to prove that we can deterministically find an independent set $\indepset$ (or a matching $\matching$) of $H$ such that in the remaining graph $H'$, $|V^{high}(H')| \leq |V^{high}(H)|/\Delta^{\Omega(1)}_{i}$ holds. To this end, we prove that the algorithm of~\Cref{lem:degree_reduction_pairwise} can be derandomized on a low-memory \mpc\ in $O(1)$ rounds. In the following, we omit specifying that $V^{low}$ and $V^{high}$ are with respect to $H$.
        We first explain how to compute a $C$-coloring of $V^{low}$ such that nodes whose random choices are assumed to be independent are given distinct colors. We color a \textit{conflict} graph on nodes in $V^{low}$ in $O(1)$ rounds constructed as follows. Let every $V^{low}$-node send its $V^{low}$-neighborhood to each of the at most $\beta^3$ nodes in $V^{high}$ at distance at most $2$. The total information exchanged in this step is at most  $\sum_{v \in V(G,\Delta_{i})'} \beta^{(1+2+2)} \le \sum_{v \in V(G,\Delta_{i})'} \deg_G(v)$. Then, every $v \in V^{high}$ adds an edge between each pair of vertices in the neighborhoods received. The number of edges that each node $v$ adds is $O(\beta^{10}) \ll \sqrt{\Delta_i} = O(\deg_G(v))$.
        Thus, a $O(\Delta_i^2 \log^2 n)$-coloring of the conflict graph can be computed in $O(1)$ rounds by~\Cref{lm:coloring}. This yields $O(\log C) = O(\delta \log \Delta_i + \log \log n)$ and thus the family of hash functions $\mathcal{H}$ used in~\Cref{lem:degree_reduction_pairwise} can be stored on a single machine. For each $v \in V^{high}$, using the exchanged information and by communicating the assigned color, the machine that $v$ is assigned to can check whether $v$ has no neighbors in $\mathcal{I}_h$ (or matching proposals) for every hash function $h$. We can then aggregate these numbers across machines for every hash function using~\Cref{lem:primitives} in $O(1)$ time and find one good function $h^* \in \mathcal{H}$, which gives the expected result.
    \end{proof}
    \begin{proof}[Proof of \Cref{lem:degree_reduction_mpc_preprocessing}]
        This lemma follows by applying \cref{lem:degree_reduction_mpc_single_round} once for each $i \in [i_{min}, i_{max}]$, which takes $O(c \log \log \Delta)$ rounds overall. 
        We define a sequence of graphs $G_{i_{min}},\ldots, G_{i_{max}}$ with $G_{i_{max}} = G$. Let $\mathcal{I}_i$ be the independent set that we obtain  by invoking \cref{lem:degree_reduction_mpc_single_round} with inputs $G_i$ and $c$.
        Let $G_{i-1} := G_i \setminus \mathcal{I}_i$ and define the independent set $\mathcal{I} = \bigcup_{i \in [i_{min}, i_{max}]} \mathcal{I}_i$.
        This finishes the description of the algorithm and we can indeed compute $\mathcal{I}$ in $O(\log \log \Delta)$ rounds deterministically on a low-memory \mpc. The procedure for computing a matching $\mathcal{M}$ is equivalent. 
    \end{proof}
    \begin{proof}[Proof of \Cref{lem:mpc_degree_reduction_multi_multi}]
        The first step of the algorithm computes a coloring of $G$ with $C = \poly(\Delta_{sup}) \cdot \log^{(k)} n$ colors such that vertices at distance at most $4$ are assigned distinct colors. This can be achieved by simulating $O(k)$ executions of \Cref{lm:coloring} in $O(1)$ \mpc\ rounds using knowledge of the $k$-hop neighborhood. In particular, this provides a proper coloring for simulating \cref{lem:degree_reduction_multi_multi} with input parameters $G$ and $\kappa \in \mathbb{N}$, where $\kappa = \Theta(k)$ is chosen such that: (i) the number of random bits required $b = O(\kappa \cdot \log (\Delta_{sup} \cdot \log^{(k)} n))$ gives $2^b = \localspace$, i.e., the $\kappa$ families of pairwise independent hash functions used by \cref{lem:degree_reduction_multi_multi} fit into the memory of a single machine; (ii) the \local\ running time $O(\kappa)$ of \Cref{lem:degree_reduction_multi_multi} is at most $k$, i.e., each node can simulate it using its $k$-hop neighborhood without any communication; and (iii) $\kappa$ is larger than the absolute constant $s$ of \Cref{lem:degree_reduction_multi_multi}. All conditions can be met at the same time using the above assumption on $k$. 
    
        We note that the following multi-stage derandomization is inspired by that of~\cite[Section 5.2]{MPC-MIS-det-general}.
        Let $\mathcal{H}$ be a $2$-wise independent hash family as explained in \Cref{lem:degree_reduction_pairwise}. By the probabilistic method, there is a sequence of hash functions $h_1, \ldots, h_{\kappa} \subset \mathcal{H}^k$ that achieves the result of \Cref{lem:degree_reduction_multi_multi}. Moreover, by the choice of $\kappa$, all possible sequences $h_1, \ldots, h_{\kappa}$ can be stored and evaluated on any single machine.
        Each high-degree node considers all possible sequences of hash functions $\mathbf{h} = \langle h_1,\ldots,h_{\kappa}\rangle$ and simulates \Cref{lem:degree_reduction_multi_multi} according to $\mathbf{h}$. Let $\indepset^{\mathbf{h}}$ be the independent set found. Each node can check whether it is incident to $\indepset^{\mathbf{h}}$ and compute its degree in $G \setminus \indepset^{\mathbf{h}}$. By aggregating these numbers for each class $V_{\ge i}(G)$, we can compute $V_{\ge i}(G \setminus \indepset^{\mathbf{h}})$ and check whether $V_{\ge i}(G') \le V_{\ge i}(G) / \Delta_i^{\Omega(i)}$ for each $i$. Thus, we select one sequence $\mathbf{h}^*$ defining $\indepset^{\mathbf{h}^*}$ which ensures that the number of nodes in each class satisfies the above condition. The same approach can be easily extended to maximal matching by running the matching version of \Cref{lem:degree_reduction_pairwise}.
    \end{proof}
    \mainsuperlinear*
    \begin{proof}
        Observe first that we can assume that $\arb = O(n^{\delta/4})$, as otherwise $\Delta \ll \poly(\arb)$ and we can thus directly apply the $O(\log \Delta + \log \log n)$-round algorithm~\cite{MPC-MIS-det-general}. Our $O(\log \log n)$-round $\poly(\arb)$-degree reduction can be applied on any graph with maximum degree $O(n^{\delta})$, for some constant $\delta \in (0,1)$. Thus, when the maximum degree of the graph under consideration is $O(n^{\delta})$, we apply our degree reduction followed by~\cite{MPC-MIS-det-general} on the remaining graph of $\poly(\arb)$-maximum degree to obtain the claimed complexity. Otherwise, our algorithm works as follows. We construct a $H$-partition of degree $d = n^{\delta/3}$ consisting of $H_1, \ldots, H_{\last}$ layers, i.e., disjoint set of vertices, with $O(\log_{d/\arb} n) = \last = O(1)$ in $O(L) = O(1)$ rounds. Recall that the maximum degree within each layer is at most $d = O(n^{\delta})$.
        For \mis, we can simply process each layer sequentially. For each $i \in [\last]$, we compute a maximal independent set $\indepset$ of the subgraph induced by $H_i$ and remove nodes in $\indepset \cup N(\indepset)$.
        For \mm, we compute a matching $\matching = \bigcup_{i=1}^{\last} \matching_i$ such that $G \setminus \matching$ contains no vertices with degree larger than $d^3$ in $H_i$, i.e., each high-degree vertex is either removed or all but at most $d^3$ of its incident edges are in $\matching$. To achieve that, we run the process explained in \Cref{sec:mmveryhighdeg}.
    \end{proof}
    \begin{lemma}\label{sec:mmveryhighdeg}
        Every vertex in $H_i$ with degree at least $d^3$ gets removed or all but at most $d^3$ of its incident edges are removed.
    \end{lemma}
    \begin{proof}
            Every node $u \in \bigcup_{j=1}^{i-1} H_j$ marks one edge $e = \{u,v\}$ with $v \in H_i \cap N(u)$ at random. Each node $v \in H_i$ selects one incident marked edge arbitrarily. We prove that this process can be derandomized to obtain a suitable matching.
        Let $v$ be a node in layer $i$ with degree at least $d^3$. Since $v$ has at most $d$ neighbors in $H_i$, there are at least $d^3 - d \gg 2d^2$ vertices in lower layers. Let $X_u$ be a family of $k$-wise independent random variables with $X_u = 1$ if $u$ marks the edge $\{u,v\}$ and $X_u = 0$ otherwise, for each $u \in U(v) \eqdef N(v) \cap \bigcup_{j=1}^{i-1} H_j$. Define $X = \sum_{u \in U(v)} X_u$ and observe that  $\E[X] \ge 2d$ since $\Pr[X_u = 1] \ge 1/d$ and $\card{U(v)} \ge 2d^2$. By applying \Cref{lem:kwise_bound} with $k = 21/\delta$, we have $\Pr[X = 0] \le 8 \left(\frac{1}{d}\right)^{k/2} \le n^{-3}$. Now, let $\mathcal{H} = \{h \colon V \mapsto [d]\}$ be a family of $k$-wise independent hash functions specified by $k\cdot(\log n + \log d) + O(1) = O(\log n)$ bits by \Cref{lemma-hash}. Each function $h \in \mathcal{H}$ defines a matching $\matching(h)$ that includes each node $v \in H_i$ with at least one marked edge, i.e., if $h(u) = v$ for $u \in U(v)$.  
        For each $v \in H_i$ of degree at least $d^3$, our derandomization scheme groups the edges of $v$ in blocks of size exactly $2d^2 = O(n^\delta)$, possibly deleting some spare edges. Each machine will now ensure that its assigned blocks are marked, i.e., each block has at least one marked edge under $h$.     
        The above concentration bound together with a union bound over the set of all blocks ensures that a randomly chosen $h$ has all blocks \textit{marked} with probability at least $n^{-1}$. Thus, each machine $x$, which is responsible for some blocks of edges $B_j(v)$ incident to node $v$, can compute locally the number of edges $\{u,v\} \in B_j(v)$ such that $h(u) = v$ for each hash function $h \in \mathcal{H}$ and determine whether $B_j(v)$ is marked. By the distributed method of conditional expectation (\Cref{sec:derandomization_framework}) with the objective function $q_{x}(h) = \mathbbm{1}\{B_j(v) \text{ is marked for every $v$}\}$, we find a function $h^*$ that makes all blocks marked in $O(1)$ time.
    \end{proof}
        
\section{Algorithm for High-Arboricity Graphs}\label{sec:higharbmis}
    \subsection{Maximal Matching Case}
    Our starting point is the maximal matching sparsification procedure of~\cite[Section 3]{MPC-MIS-det-general}, which is summarized in the following lemma and works with $O(n+m)$ total space.
    \begin{lemma}[\cite{MPC-MIS-det-general}]\label{lm:spar_match}
        There is a deterministic low-memory \mpc\ constant-round sparsification procedure that returns a subset of nodes $B \subseteq V$ and a subset of edges $E' \subseteq E$ such that:
        \begin{enumerate}
            \item For every $v \in V$ it holds that $\deg_{E'}(v) = O(n^\delta)$, for any $\delta$, $0 < \delta < 1$. \label{prop:spar_match1}
            \item Every node $v \in B$ either satisfies $\sum_{\{u,v\} \in E'} \frac{1}{\deg_{E'}(\{u,v\})} \ge \frac{1}{27}$, or is incident to an edge $\{u,v\} \in E'$ whose degree in $E'$ is $0$. \label{prop:spar_match2}
            \item For the subset $B$ of $V$ it holds that $\sum_{v \in B} \deg(v) \ge \frac{\delta |E|}{8}$. \label{prop:spar_match3}
        \end{enumerate}
    \end{lemma}
    
    We use the subsets $E'$ of edges and $B$ of vertices returned by \Cref{lm:spar_match} to construct a matching $\matching \subseteq E'$ that is incident to $\Omega(\delta |E|)$ edges. We define for each $v \in B$ a subset of \textit{relevant} edges $S(v) \subseteq E'(v)$ such that $\sum_{e \in S(v)} \frac{1}{\deg_{E'}(\{u,v\})} \in \left[\frac{1}{27}, 1\right]$, by Property~\ref{prop:spar_match2}. 
    Let $\mathcal{H}$ be a family of pairwise independent hash functions and let $h \in \mathcal{H}$ map each edge $e = \{u,v\}$ in $E'$ to a value $z_e \in [n^3]$. We say that $e$ is \textit{sampled} and joins the set of sampled edges $S_h$ iff $z_e < \frac{n^3}{3\deg_{E'}(e)}$. 
    Let $S_h(v) \eqdef S(v) \cap S_h$ denote the set of edges in $S(v)$ that are sampled. Define $\matching_{h}$ to be the set of sampled edges with no neighboring edges sampled. Observe that, for a node $v$, if the following two conditions hold: \textit{(i)} $\card{S_h(v)} = 1$ and \textit{(ii)} $e = \{u,v\} = \{e' \in S_h \mid u \in e'\}$, i.e., $e$ is the only sampled edge in $u$'s neighborhood, then $v$ is matched under $h$. For simplicity, let us remove from consideration all nodes $v \in B$ with $S(v) = \{e\}$. 
    We now define the following pessimistic estimator $P(h)$:
    \begin{align*}
        P(h) &= \sum_{v \in B \colon \card{S_h(v)} = 1} \deg(v) - \sum_{v \in B} \deg(v) \cdot
          \left(\sum_{\{u,v\} \in S_h(v)} \mathbbm{1}\{\exists_{e' \in S_h \sim e} e' \ni u\}\right).
    \end{align*}
    To compute $P(h)$, every $v \in B$ checks $\card{S_h(v)}$ and, for each edge $e = \{u,v\} \in S_h(v)$, node $u$ checks whether there is a $e' = \{u,w\} \in S_h$ with $w \neq v$. This shows that $P(h)$ can be computed by storing the one-hop neighborhoods of vertices on single machines.  We now lower bound $\E[P(h)]$. Let $A_e$ be the event $e \in S_h$ and $B_e$ be the bad event $\left\{ \exists_{e' \in E' \sim e} A_{e'} \right\}$ that $e$ is neighboring a sampled edge.
    \begin{align*}
        \E[P(h)] &= \sum_{v \in B} \deg(v) \cdot \Pr\left[\card{S_h(v)} = 1\right] - \sum_{v \in B} \deg(v) \cdot \left(\sum_{e \in S(v)} \Pr\left[A_e \land B_e \right]\right)\\
        &= \sum_{v \in B} \deg(v) \cdot \left(\Pr\left[\card{S_h(v)} = 1\right] -  \sum_{e \in S(v)} \Pr\left[A_e \land B_e \right]\right).
    \end{align*}
    The probability of the event $A_e$ is
    $
        \frac{1}{3\deg_{E'}(e)} - \frac{1}{n^3} \le \Pr\left[A_e\right] \le \frac{1}{3\deg_{E'}(e)}.
    $
    Conditioned on $A_e$ and using pairwise independence of $z_e$ and $z_{e'}$, we obtain
    \begin{align*}
        \Pr\left[ B_e \mid A_e\right] &\le \sum_{e' \in E' \sim e} \Pr\left[A_{e'}\right] 
        %\le \sum_{e' \in E' \sim e} \Pr\left[z_{e'} < \frac{n^3}{3\deg_{E'}(e)}\right] \\
        \le \deg_{E'}(e) \cdot \frac{1}{3\deg_{E'}(e)} = \frac 1 3.
    \end{align*}
    Then, by applying the inclusion-exclusion principle and noticing that $\Pr[A_e, A_{e'}] = \Pr[A_e] \cdot \Pr[A_{e'}]$ by pairwise independence, we get
    \begin{align*}
        \Pr\left[\card{S_h(v)} = 1\right] 
        %&\ge \sum_{e \in S(v)} \Pr\left[A_e\right] - \sum_{\substack{e, e' \in S(v),\, e \neq e'}} \Pr\left[A_e, A_{e'}\right]\\
        &\ge \sum_{e \in S(v)} \Pr\left[A_e\right] - \sum_{e \in S(v)} \sum_{e' \in S(v)} \Pr\left[A_e\right] \cdot \Pr\left[A_{e'}\right]\\
        &\ge \sum_{e \in S(v)} \Pr\left[A_e\right] \left(1 - \sum_{e' \in S(v)} \Pr\left[A_{e'}\right]\right). 
    \end{align*}
    Observe that $\Pr\left[A_e \land B_e \right] = \Pr\left[A_e\right] \Pr\left[B_e \mid A_e\right] \le \frac{1}{3} \Pr\left[A_e\right]$ by the calculations above. Thus,
    \begin{align*}
        & \Pr\left[\card{S_h(v)} = 1\right] - \sum_{e \in S(v)} \Pr\left[A_e \land B_e \right] 
       \\& = \sum_{e \in S(v)} \Pr\left[A_e\right] \left(\frac{2}{3} - \sum_{e' \in S(v)} \Pr\left[A_{e'}\right]\right)\\
        &\ge \sum_{e \in S(v)} \Pr\left[A_e\right] \left(\frac{2}{3} - \frac 1 3 \sum_{e' \in S(v)} \frac{1}{\deg_{E'}(e')}\right) \ge \frac{1}{3} \cdot \sum_{e \in S(v)} \Pr\left[A_e\right]\\
        &\ge \frac{1}{3} \cdot \sum_{e \in S(v)} \left(\frac{1}{3\deg_{E'}(e)} - \frac{1}{n^3}\right)
    \\ & \ge \frac{1}{9} \cdot \sum_{e \in S(v)} \frac{1}{\deg_{E'}(e)} - \frac{1}{3n^2} 
        \ge \frac{1}{243} - \frac{1}{3n^2},
    \end{align*}
    where the last inequality follows from Property~\ref{prop:spar_match2}. Note that we can assume $\frac{1}{243} - \frac{1}{3n^2} \ge \frac{1}{250}$ for $n$ large enough. This yields $\E[P(h)] \ge \frac{1}{250} \sum_{v \in B} \deg(v) \ge \frac{\delta |E|}{2000}$, as desired.
    Summarized, our algorithm applies the sparsification procedure of \Cref{lm:spar_match} followed by the  derandomization approach of \Cref{sec:derandomization_framework} with objective function $P(h)$. By the method of conditional expectation, we select a hash function $h^*$ such that $P(h^*) \ge \frac{\delta}{2000} \card{E}$, i.e., a matching that decreases the number of edges by $\frac{\delta}{4000} \card{E}$. Repeating it $O(\log n)$ times ensures the matching maximality.
    \subsection{MIS Case}
    \begin{lemma}[MIS sparsification {\cite[Section 4]{MPC-MIS-det-general}}]\label{lm:spar_mis}
        There is a deterministic low-memory \mpc\ constant-round sparsification procedure that returns subsets of nodes $B, Q \subseteq V$ and a subset of edges $E' \subseteq E$ such that:
        \begin{enumerate}
            \item For every $v \in B \cup Q$ it holds that $\deg_{E'}(v) = O(n^\delta)$ and that $\deg_{Q}(v) = O(n^\delta)$, for any $\delta$, $0 < \delta < 1$. \label{prop:spar_mis1}
            \item Every node $v \in B$ either satisfies $\sum_{u \in N_{E'}(v)} \frac{1}{\deg_{Q}(u)} \ge \frac{\delta}{10}$, or it has a neighbor $u \in Q$ whose degree in $Q$ is $0$. \label{prop:spar_mis2}
            \item For the subset $B$ of $V$ it holds that $\sum_{v \in B} \deg(v) \ge \frac{\delta |E|}{8}$. \label{prop:spar_mis3}
        \end{enumerate}
    \end{lemma}
   
    We use the subsets $E'$ of edges and $B,Q$ of vertices returned by \Cref{lm:spar_mis} to find an independent set $\indepset \subseteq Q$ such that $\indepset \cup N(\indepset)$ is incident to a linear number of edges. For each $v \in B$, we define a subset of \textit{relevant} nodes by $S(v) \subseteq N_{E'}(v)$, with $\sum_{u \in S(v)} \frac{1}{\deg_{Q}(u)} \in [\delta/10, 1]$ by Property \ref{prop:spar_mis2}. 
    Let $\mathcal{H}$ be a family of pairwise independent hash functions and let $h \in \mathcal{H}$ map each node $v$ in $Q$ to a value $z_v \in [n^3]$. We say that $v$ is \textit{sampled} and joins the set of sampled nodes $S_h$ iff $z_v < \frac{n^3}{3\deg_{Q}(v)}$. Let $S_h(u) \eqdef S(u) \cap S_h$ denote the set of nodes in $S(u)$ that are sampled for each $u \in B$. 
    Define $\indepset_{h}$ to be the set of sampled nodes such that no neighbor is sampled. Observe that, for a node $v \in B$, if the following two conditions hold: \textit{(i)} $S_h(v) = \{u\}$ and \textit{(ii)} $S_h \cap N_{Q}(u) = \emptyset$, then $v$ is incident to $\indepset_{h}$. For simplicity, let us remove from consideration all vertices adjacent to a vertex $u$ with degree $0$ in $Q$. To find a good $\indepset_{h}$, we define the following pessimistic estimator $P(h)$:
    \begin{align*}
        P(h) &= \sum_{v \in B \colon \card{S_h(v)} = 1} \deg(v) - \sum_{v \in B} \deg(v) \cdot \left(\sum_{u \in S_h(v)} \mathbbm{1}\{\exists_{u' \in N_Q(u)} u' \in S_h\}\right).
    \end{align*}
    To compute $P(h)$, each node $v \in B$ checks $\card{S_h(v)}$ and, for each edge $e = \{u, v\}$ with $u \in S_h(v)$, node $u$ checks whether $u$ is neighboring a sampled node. Let $A_u$ be the event that $u \in S_h$ and $B_u$ denote the bad event $\left\{ \exists_{u' \in N_Q(u)} u' \in S_h \right\}$ that $u$ neighbors a sampled node. We have
    \begin{align*}
        \E[P(h)] &= \sum_{v \in B} \deg(v) \cdot \Pr\left[\card{S_h(v)} = 1\right] - \sum_{v \in B} \deg(v) \cdot \left(\sum_{u \in S(v)} \Pr\left[A_u \land B_u \right]\right)\\
        &= \sum_{v \in B} \deg(v) \cdot \left(\Pr\left[\card{S_h(v)} = 1\right] -  \sum_{e \in S(v)} \Pr\left[A_u \land B_u \right]\right).
    \end{align*}
    The probability of the event $A_u$ is
        $
            \frac{1}{3\deg_{Q}(u)} - \frac{1}{n^3} \le \Pr\left[A_u\right] \le \frac{1}{3\deg_{Q}(u)}
        $.
    Conditioned on $A_u$ and by independence of $z_u, z_{u'}$, the probability of $B_u$ is         
    \begin{align*}
            \Pr\left[B_u \mid A_u\right] &\le \sum_{u' \in N_Q(u)} \Pr\left[A_{u'}\right] \le \deg_{Q}(u) \cdot \frac{1}{3\deg_{Q}(u)} = \frac 1 3.
    \end{align*}
    Then, by applying the inclusion-exclusion principle and observing that $\Pr\left[A_u, A_{u'}\right] = \Pr\left[A_u\right] \cdot \Pr\left[A_{u'}\right]$ by pairwise independence, we get
    \begin{align*}
        \Pr\left[\card{S_h(v)} \ge 1\right] &\ge \sum_{u \in S(v)} \Pr\left[A_u\right] - \sum_{\substack{u, u' \in S(v),\\ u \neq u'}} \Pr\left[A_u\right] \cdot \Pr\left[A_{u'}\right]\\
        &\ge \sum_{u \in S(v)} \Pr\left[A_u\right] \left(1 - \sum_{u' \in S(v)} \Pr\left[A_{u'}\right]\right). 
    \end{align*}
    Observe that $\Pr\left[A_u \land B_u \right] = \Pr\left[A_u\right] \cdot \Pr\left[B_u \mid u \in S_h\right] \le \frac 1 3 \Pr\left[A_u\right]$ by the calculations above.
    Thus,
    \begin{align*}
        &\Pr\left[\card{S_h(v)} = 1\right] - \sum_{u \in S(v)} \Pr\left[A_u \land B_u \right] \\ & \ge \sum_{u \in S(v)} \Pr\left[A_u\right] \left(\frac{2}{3} - \sum_{u' \in S(v)} \Pr\left[A_{u'}\right]\right)\\
        &= \sum_{u \in S(v)} \Pr\left[A_u \right] \left(\frac{2}{3} - \frac 1 3 \sum_{u' \in S(v)} \frac{1}{\deg_{Q}(u')}\right) \ge \frac{1}{3} \cdot \sum_{u \in S(v)} \Pr\left[A_u\right]\\
        &\ge \frac{1}{3} \cdot \sum_{u \in S(v)} \left(\frac{1}{3\deg_{Q}(u)} - \frac{1}{n^3}\right)
        \ge \frac{1}{9} \cdot \sum_{u \in S(v)} \frac{1}{\deg_{Q}(u)} - \frac{1}{3n^2} \\ &\ge \frac{\delta}{90} - \frac{1}{3n^2},
    \end{align*}
    where the last inequality follows from Property~\ref{prop:spar_match2}. Note that we can assume $\frac{\delta}{90} - \frac{1}{3n^2} \ge \frac{\delta}{100}$ for $n$ large enough. We can now conclude that $\E[P(h)] \ge \frac{\delta}{100} \sum_{v \in B} \deg(v) \ge \frac{\delta^2 |E|}{800}$, as desired.
    Summarized, our algorithm applies the sparsification procedure of \Cref{lm:spar_mis} followed by the above derandomization process with objective function $P(h)$. We can thus find a hash function $h^*$ such that $P(h^*) \ge \frac{\delta^2}{800}\card{E}$, i.e., an independent set $\indepset_{h^*}$ incident to at least $\frac{\delta}{1600} \card{E}$ edges. Thus, after $O(\log n)$ iterations, the returned independent set is maximal.
\end{document}